\documentclass[12pt]{article}
\usepackage{fullpage}
\usepackage{mathrsfs}
\usepackage{float}
\usepackage{amsmath, amsfonts, amssymb, amsthm, amsbsy, amscd, bm, bbm}
\usepackage{array}
\usepackage{booktabs}
\usepackage{graphicx}
\usepackage[small,bf]{caption}
\usepackage{subcaption}
\usepackage{color}
\usepackage{ifthen}
\usepackage{xspace}
\usepackage{algorithm}
\usepackage[colorlinks,citecolor={black},urlcolor={black},linkcolor={black}]{hyperref}
\usepackage{url}
\usepackage{tocbibind}

\usepackage[noend]{algpseudocode}

\newtheorem{theorem}{Theorem}[section]
\newtheorem{corollary}[theorem]{Corollary}

\newtheorem{definition}[theorem]{Definition}
\newtheorem{lemma}[theorem]{Lemma}

\newtheorem{claim}{Claim}[section]

\def\E{\mathbb{E}}

\title{A Nearly Instance Optimal Algorithm for Top-$k$ Ranking under the Multinomial Logit Model}

\author{
Xi Chen \thanks{Stern School of Business, New York University, email: xchen3@stern.nyu.edu}
\and
Yuanzhi Li \thanks{Department of Computer Science, Princeton University, email: yuanzhil@cs.princeton.edu}
\and
Jieming Mao  \thanks{Department of Computer Science, Princeton University, email: jiemingm@cs.princeton.edu}
}
\begin{document}
\maketitle

\begin{abstract}
We study the active learning problem of top-$k$ ranking from multi-wise comparisons under the popular multinomial logit model. Our goal is to identify the top-$k$ items with high probability by adaptively querying sets for comparisons and observing the noisy output of the most preferred item from each comparison. To achieve this goal, we design a new active ranking algorithm without using any information about the underlying items' preference scores. We also establish a matching lower bound on the sample complexity even when the set of preference scores is given to the algorithm. These two results together show that the proposed algorithm is nearly instance optimal (similar to instance optimal \cite{FaginLN03}, but up to polylog factors). Our work extends the existing literature on rank aggregation in three directions. First, instead of studying a static problem with fixed data, we investigate the top-$k$ ranking problem in an active learning setting. Second, we show our algorithm is nearly instance optimal, which is a much stronger theoretical guarantee. Finally, we extend the pairwise comparison to the multi-wise comparison, which has not been fully explored in ranking literature.

%Our goal is to develop an algorithm, which by adaptively querying sets for comparisons, identifies the top-$k$ items with high probability.
%we design a new active ranking algorithm without using any information about the noise parameters. We also establish a matching lower bound (up to logarithmic factors) on the sample complexity for the problem even when the noise parameters are given to the algorithm. These two results together show that the proposed algorithm is instance optimal (up to logarithmic factors). Our work extends the existing literature on rank aggregation in three directions. First, instead of studying a static problem with fixed data, we investigate the top-$k$ ranking problem in an \emph{active} learning setting. Second, we consider the instance optimality which we think is appropriate under this noise model. Third, we extend the pairwise comparison to multi-wise comparison model, which has not been fully explored in ranking literature.
\end{abstract}

\section{Introduction}

The problem of inferring a ranking over a set of $n$ items (e.g., products, movies, URLs) is an important problem in machine learning and finds numerous applications in recommender systems, web search, social choice, and many other areas. To learn the global ranking, an effective way is to present at most $l$ ($l\geq 2$) items at each time and ask about the most favorable item among the given items. Then, the answers from these  multi-wise comparisons will be aggregated to infer the global ranking. When the number of items $n$  becomes large, instead of inferring the global ranking over all the $n$ items, it is of more  interest to identify the top-$k$ items with a pre-specified $k$. In this paper, we study the problem of active top-$k$ ranking from multi-wise comparisons, where the goal is to adaptively choose at most $l$ items for each comparison and accurately infer the top-$k$ items with the minimum number of comparisons (i.e., the minimum sample complexity). As an illustration, let us consider a practical scenario: an online retailer is facing the problem of choosing $k$ best designs of handbags among $n$ candidate designs. One popular way is to display several designs to each arriving customer and observe which handbag is chosen. Since a shopping website has a capacity on the maximum number of display spots, each comparison will involve at most $l$ possible designs.

%In this paper, we study the problem of active learning of top-$k$ items from multi-wise comparisons. AThe goal is to adaptively choose at most $l$ items for comparison at each time and infer the top-$k$ items with the minimum number of comparisons (i.e., the minimum sample complexity).

Given the wide application of top-$k$ ranking, this problem has received a lot of attention in recent years, e.g., \cite{Shah15Sim, Suh16Adversarial} (please see Section \ref{sec:related} for more details). Our work greatly extends the existing literature on top-$k$ ranking in the following three directions:
\begin{enumerate}
  \item Most existing work studies a non-active ranking aggregation problem, where the answers of comparisons are provided statically or the items for each comparison are chosen completely at random. Instead of considering a passive ranking setup, we propose an active ranking algorithm, which adaptively chooses the items for comparisons based on previously collected information.
  %Moreover, our sample complexity bound is instance-adaptive, which is a function involving not only the problem parameters $n$, $k$, and $l$, but also the underlying parameters that quantify the preference/utility of each item.
 \item Most existing work chooses some specific function (call this function $f$)  of problem parameters (e.g., $n$, $k$, $l$ and preference scores) and shows that the algorithm's sample complexity is at most $f$. For the optimality, they also show that for any value of $f$, there exists an instance whose sample complexity equals to that value and any algorithm needs at least $\Omega(f)$ comparisons on this instance. However, this type of algorithms could perform poorly on some instances other than those instances for establishing lower bounds (see examples from \cite{Chen:17:SODA}); and  the form of function $f$ can vary the designed algorithm a lot.
     %To address this issue, in this work, we compare our algorithm (which does not know the problem parameters) to the optimal algorithm when the problem are given.

     %{\red To address this issue, we establish a much more refined upper bound on the  sample complexity. The derived sample complexity is function of not only the basic problem parameters such as $n$, $k$ and $l$ but also the underlying preference scores for each item, which reflect the difficulty of each instance but are unknown to the algorithm.  To establish the instance optimality, we compare our algorithm to the optimal ``oracle algorithm'' when all the underlying preference scores are given.}
    To address this issue, we establish a much more refined upper bound on the sample complexity. The derived sample complexity matches the lower bound when all the parameters (including the set of underlying preference scores for items) are given to the algorithm. They together show that our lower bound is tight and also our algorithm is nearly instance optimal (see Definition \ref{def:io} for the definition of nearly instance optimal).

  \item Existing work mainly focuses on pairwise comparisons. We extend the pairwise comparison to the multi-wise comparison (at most $l$ items) and further quantify the role of $l$ in the sample complexity. From our sample complexity result (see Section \ref{sec:result}), we show that the pairwise comparison could be as helpful as multi-wise comparison unless the underlying instance is very easy.
 \end{enumerate}

\subsection{Model}
\label{sec:model}

In this paper, we adopt the widely used multinomial logit (MNL) model \cite{Luce59, McFadden:73, Train:03:choice} for modeling multi-wise comparisons. In particular, we assume that each item $i$ has an underlying preference score (a.k.a. utility in economics)  $\mu_i$ for $i=1,\ldots, n$. These scores, which are unknown to the algorithm, determine the underlying ranking of the items. Specifically, $\mu_i >\mu_j$ means that item $i$ is preferred to item $j$ and item $i$ should have a a higher rank. Without loss of generality, we assume that $\mu_1 \geq \mu_2 \geq \cdots \mu_k > \mu_{k+1} \geq \cdots \geq \mu_n$, and thus the true top-$k$ items are $\{1,\dots, k\}$. At each time $t$ from $1$ to $T$, the algorithm chooses a subset of items with at least two items, denoted by $S_t \subseteq \{ 1,...,n\}$, for query/comparison. The size of the set $S_t$ is upper bounded by a pre-fixed parameter $l$, i.e., $2 \leq |S_t| \leq l$.

Given the set $S_t$, the agent will report her most preferred item $a \in S_t$ following the multinomial logit (MNL) model:
\begin{equation}\label{eq:mnl}
\Pr[a |S_t] = \frac{\exp(\mu_a)}{\sum_{j\in S_t} \exp(\mu_j)}.
\end{equation}
When the size of $S_t$ is $2$ (i.e., $l=2$), the MNL model reduces to Bradley-Terry model \cite{Bradley52}, which has been widely studied in rank aggregation literature in machine learning (see, e.g., \cite{Negahban12RankCentrality,Jang13,Rajkumar14,Chen15}).

In fact, the MNL model has a simple probabilistic interpretation as follows \cite{Train:03:choice}.  Given the set $S_t$, the agent draws her valuation $\nu_j =\mu_j + \epsilon_j$ for each item $j \in S_t$, where $\mu_j$ is the mean utility for item $j$ and each $\epsilon_j$ is independently, identically distributed random variable following the Gumbel distribution. Then, the probability that $a \in S_t$ is chosen as the most favorable item is $\Pr \left(\nu_a \geq \nu_j, \forall j \in S_t\backslash \{a\} \right)$. With some simple algebraic derivation using the density of Gumbel distribution (see Chapter 3.1 in \cite{Train:03:choice}), the choice probability $\Pr \left(\nu_a \geq \nu_j, \forall j \in S_t\backslash \{a\} \right)$ has an explicit expression in \eqref{eq:mnl}.
For notational convenience, we define $\theta_j = \exp(\mu_j)$ for $i=1,...,n$, and the choice probability in \eqref{eq:mnl} can be equivalently written as $\Pr[a |S_t] = \frac{\theta_a}{\sum_{j\in S_t} \theta_j}$. By adaptively querying the set $S_t$ for $1 \leq t \leq T$ and observing the reported most favorable item in $S_t$, the goal is to identify the set of top-$k$ items with high probability using the minimum number of queries.

For notation convenience, we assume the $i$-th item (with the preference score $\theta_i$) is labeled as $\pi_i \in \{1,\ldots, n\}$ by the algorithm at the beginning. Since the algorithm has no prior knowledge on the ranking of items before it makes any comparison, the ranking of the items should have no correlation with the labels of the items. Therefore, $\pi = (\pi_1,...,\pi_n)$ is distributed as a uniform permutation of $\{1,...,n\}$.
%\xnote{Jieming: this paragraph still sounds a bit strange to me. Should this technical detail be deferred to a later place instead of there?}
%\jnote{To Xi: I think it is important notation. It might look strange as previously it is written as item $i$ ranked as $\pi_i$, I changed it to item $i$ labeled as $\pi_i$. It should not be hard to understand in the following sense: suppose you want to run a ranking algorithm, and you don't know the ranking, then you still want to give the items some names. (Otherwise how do you remember what's going on.) Those names are just $\pi_1,...,\pi_n$.}

 The notion of instance optimal was originally defined and emphasized as an important concept in \cite{FaginLN03}. With the MNL model in place, we provide a formal definition of nearly instance optimal in our problem. To get a definition of instance optimal in our problem, we can just replace $\tilde{O}$ with $O$ in Definition \ref{def:io}. The ``nearly'' here just means we allow polylog factors.
\begin{definition}[Nearly Instance Optimal]
\label{def:io}
Given instance $(n,k, l, \theta_1,...,\theta_n)$, define \\$c(n,k,l,\theta_1,...,\theta_n)$ to be the sample complexity of an optimal adaptive algorithm on the instance. We say that an algorithm $A$ is nearly instance optimal, if for any instance $(n,k,l,\theta_1,...,\theta_n)$, the algorithm $A$ outputs the top-$k$ items with high probability and only uses at most \\$\tilde{O}(c(n,k,l,\theta_1,...,\theta_n))$ number of comparisons.  (Note that $\tilde{O}(\cdot)$ hides polylog factors of $n$ and  $\frac{1}{\theta_k - \theta_{k+1}}$.)
\end{definition}

\subsection{Main results}
\label{sec:result}

Under the MNL model described in Section \ref{sec:model}, the main results of this paper include the following upper and lower bounds on the sample complexity.

\begin{theorem}\label{thm:upper}
We design an active ranking algorithm which uses
\[
\tilde{O}  \Biggl( \frac{n}{l}  + k + \frac{\sum_{i \geq k  + 1} \theta_{i}}{\theta_k} +  \sum_{i \geq k + 1, \theta_i \geq \frac{\theta_k}{2}} \frac{\theta_{k}^2}{(\theta_k - \theta_i)^2}   +\sum_{ i: i\leq k, \theta_i \leq 2\theta_{k+1}} \frac{\theta_{k+1}^2}{(\theta_{ k + 1} - \theta_i)^2}  \Biggr)
\]
comparisons with the set size at most $l$ (can be $2$-wise, $3$-wise,...,$l$-wise comparisons) to identify the top-$k$ items with high probability. %(The algorithm does not have to know $\{\theta_1,...,\theta_n\}$.)
\end{theorem}

We note that in Theorem \ref{thm:upper}, the notation $\tilde{O}(\cdot)$ hides polylog factors of $n$ and  $\frac{1}{\theta_k - \theta_{k+1}}$.

Next, we present a matching lower bound result, which shows that our sample complexity in Theorem \ref{thm:upper} is nearly instance optimal.
\begin{theorem}\label{thm:lower}
For any (possibly active) ranking algorithm $A$, suppose that $A$ uses comparisons of set size at most $l$. Even when the algorithm $A$ is given the values of $\{\theta_1,...,\theta_n\}$ (note that $A$ does not know which item takes the preference score $\theta_i$ for each $i$),  $A$ still needs
\[
\Omega  \Biggl( \frac{n}{l}  + k + \frac{\sum_{i \geq k  + 1} \theta_{i}}{\theta_k} +  \sum_{i \geq k + 1, \theta_i \geq \frac{\theta_k}{2}} \frac{\theta_{k}^2}{(\theta_k - \theta_i)^2}   +\sum_{ i: i\leq k, \theta_i \leq 2\theta_{k+1}}\frac{\theta_{k + 1 }^2}{(\theta_{k+1} - \theta_i)^2}  \Biggr)
\]
comparisons to identify the top-$k$ items with probability at least $7/8$.
\end{theorem}

Here we give some intuitive explanations of the terms in the above bounds before introducing the proof overview:
\begin{enumerate}
\item Term $\frac{n}{l}$: Since each comparison has size at most $l$, we need at least $\frac{n}{l}$ comparisons to query each item at least once.
\item  Term  $k$: As the proof will suggest, in order to find the top-$k$ items, we need to observe most items in the top-$k$ set as chosen items from comparisons.
    However, we do not have to observe most items in the bottom-$(n-k)$ set. Therefore, there is no term $n-k$ in the bound.
\item  Term  $\frac{\sum_{i \geq k  + 1} \theta_{i}}{\theta_k} +  \sum_{i \geq k + 1, \theta_i \geq \frac{\theta_k}{2}} \frac{\theta_{k}^2}{(\theta_k - \theta_i)^2}   +\sum_{ i: i\leq k, \theta_i \leq 2\theta_{k+1}} \frac{\theta_{k+1}^2}{(\theta_{ k + 1} - \theta_i)^2}$:
     Roughly speaking,   when $i>k$ and $\theta_i \geq \theta_k/2$, $\Theta\left( \frac{(\theta_k - \theta_i)^2}{\theta_{k}^2} \right)$ is the amount of information that the comparison between item $i$ and item $k$ reveals. So intuitively, we need $\Omega\left( \frac{\theta_{k}^2}{(\theta_k - \theta_i)^2} \right)$ to tell that item $i$ ranks after item $k$. Other quantity can also be understood from an information theoretic perspective.
\end{enumerate}

It is also worthwhile to note that when $l$ is a constant, it's easy to check that
\begin{eqnarray*}
&&\frac{n}{l}  + k + \frac{\sum_{i \geq k  + 1} \theta_{i}}{\theta_k} +  \sum_{i \geq k + 1, \theta_i \geq \frac{\theta_k}{2}} \frac{\theta_{k}^2}{(\theta_k - \theta_i)^2}   +\sum_{ i: i\leq k, \theta_i \leq 2\theta_{k+1}}\frac{\theta_{k+1}^2}{(\theta_{ k + 1} - \theta_i)^2} \\
&=& O\Biggl(\sum_{i =k+1}^n\frac{\theta_{k}^2}{(\theta_k - \theta_i)^2} + \sum_{ i=1}^k\frac{\theta_i^2}{(\theta_{ k + 1} - \theta_i)^2}\Biggr).
\end{eqnarray*}
This is a simpler expression of the instance optimal sample complexity when $l$ is a constant.

Based on the sample complexity results in Theorem \ref{thm:upper} and \ref{thm:lower}, we summarize the main theoretical contribution of this paper:
\begin{enumerate}
\item We design an active ranking algorithm  for identifying top-$k$ items under the popular MNL model.
      We further prove a matching lower bound, which establishes that the proposed algorithm is nearly instance optimal.
\item Our result shows that the improvement of the multi-wise comparison over the pairwise comparison depends on the difficulty of the underlying instance. Note that the only term in the sample complexity involving $l$ is $\frac{n}{l}$. Therefore, the multi-wise comparison makes a significant difference from the pairwise comparison only when $\frac{n}{l}$ is the leading term in the sample complexity.

    Therefore,  unless the underlying instance is really easy (e.g., the instance-adaptive term  $k+\frac{\sum_{i \geq k  + 1} \theta_{i}}{\theta_k} +  \sum_{i \geq k + 1, \theta_i \geq \frac{\theta_k}{2}} \frac{\theta_{k}^2}{(\theta_k - \theta_i)^2}   +\sum_{ i: i\leq k, \theta_i \leq 2\theta_{k+1}}  \frac{\theta_{k + 1 }^2}{(\theta_{ k + 1} - \theta_i)^2}$ is $o(n)$. One implication is that most of the $\theta_i$'s among $\theta_{k+1},...,\theta_n$ are much smaller than $\theta_k$), the pairwise comparison is as helpful as the multi-wise comparison.
\end{enumerate}

\subsection{Proof overview}
In this section, we give some very high level overviews of how we prove Theorem \ref{thm:upper} and Theorem \ref{thm:lower}.

\subsubsection{Algorithms}
To prove Theorem \ref{thm:upper}, we consider two separate cases: $l= O(\log n)$ or $l = \Omega(\log n)$.
\begin{enumerate}
\item
In the first case, by losing a log-factor, we can just focus on only using pairwise comparisons. Our algorithm first randomly select $\tilde{O}(n)$ pairs and proceed by querying all of them once per iteration. After getting the query results, by a standard binomial concentration bound, we are able to construct a confident interval of $\frac{\theta_i}{ \theta_j}$ for each pair $(i, j)$ selected by the algorithm in the beginning. In a high level, our algorithm goes by declaring $\theta_i \geq \theta_j $ for pair $i, j$, if the lower bound of the corresponding confident interval is bigger or equal to $1$, or if there already exists $d$ items $(i = i_1), i_2, \ldots, (i_d = j)$ such that we have already declared $\theta_{i_r} \geq \theta_{i_{r + 1}}$ for all $r \in [d - 1]$. We are able to show that, if $\theta_1 \geq \theta_2 \geq \cdots \geq \theta_n$, then for all $i, j \in [n]$ with $j \geq i + \frac{n}{4}$, the algorithm will successfully declare $\theta_i \geq \theta_j$ after $O\left( \sum_{i =k+1}^n\frac{\theta_{k}^2}{(\theta_k - \theta_i)^2} + \sum_{ i=1}^k\frac{\theta_i^2}{(\theta_{ k + 1} - \theta_i)^2} \right)$ many total queries. Thus, we can remove at least $\frac{n}{4}$ items and recurse on a  smaller set.

\item The more interesting case is when $l = \Omega(\log n)$. As we have argued before, it is only beneficial to use multi-wise comparisons when $k+\frac{\sum_{i \geq k  + 1} \theta_{i}}{\theta_k} +  \sum_{i \geq k + 1, \theta_i \geq \frac{\theta_k}{2}} \frac{\theta_{k}^2}{(\theta_k - \theta_i)^2}   +\sum_{ i: i\leq k, \theta_i \leq 2\theta_{k+1}}  \frac{\theta_{k + 1 }^2}{(\theta_{ k + 1} - \theta_i)^2} = o(n)$. This implies that $\frac{\sum_{i \geq k  + 1} \theta_{i}}{\theta_k} = o(n-k)$ and therefore among $\theta_{k+1},...,\theta_{n}$, there are more than half of $\theta_i$'s whose value is smaller than some constant fraction of $\theta_k$. Thus, intuitively, if we select a random subset of items that contains $\theta_k$ and keep querying this set, then,  instead of seeing all items in this set with roughly equal probability, we will be seeing item $k$ much more often than the median of frequencies of items in the set. Thus, our algorithm can select an item if it ``appears very often when querying a set containing it''. We will show that, if the number of total queries is
\[
\Omega\Biggl(\frac{n}{l}+ k+\frac{\sum_{i \geq k  + 1} \theta_{i}}{\theta_k} +  \sum_{i \geq k + 1, \theta_i \geq \frac{\theta_k}{2}} \frac{\theta_{k}^2}{(\theta_k - \theta_i)^2}   +\sum_{ i: i\leq k, \theta_i \leq 2\theta_{k+1}}  \frac{\theta_{k + 1 }^2}{(\theta_{ k + 1} - \theta_i)^2} \Biggr),
\]then we will be able to select all the top $k$-items while not selecting any of the bottom $n/2$ items.   Thus, we can remove at least $\frac{n}{2}$ items and recurse on a smaller set.
\end{enumerate}

\subsubsection{Lower bounds}
To prove Theorem \ref{thm:lower}, we establish several lower bounds and combine them using a simple averaging argument. Most of these lower bounds follow the following general proof strategy:
\begin{enumerate}
\item For a given instance $(n,k,l,\theta_1,...,\theta_n)$, consider other instances on which no algorithm can output $\{\pi_1,...,\pi_k\}$ with high probability \footnote{Recall that $\pi_i$ denotes the initial label of $i$-th item given as the input to the algorithm, and thus the true top-$k$ items are labeled by  $\{\pi_1,...,\pi_k\}$.}.
    For example if we just change $\theta_{k+1}$ to $\theta_k$, then no algorithm can output $\{\pi_1,...,\pi_k\}$ with probability more than $1/2$. This is because item $k$ and item $k+1$ look the same now and thus all the algorithms will output $\{\pi_1,...,\pi_k\}$ and $\{\pi_1,...,\pi_{k-1},\pi_{k+1}\}$ with the same probability in the modified instance.
\item We then consider a well-designed distribution over these modified instances. We show that for any algorithm $A$ with not enough comparisons, the transcript of running $A$ on the original instance distributes very closely to the transcript of running $A$ on the well-design distribution over modified instances.
\item Finally, since the transcript also includes the output, step 2 will tell us that if $A$ does not use enough comparisons, then $A$ must fail to output $\{\pi_1,...,\pi_k\}$ with some constant probability.
\end{enumerate}

%\xnote{Is  $\{\pi_1,...,\pi_k\}$ the right ranking top-$k$ items? Do you want to add a sentence to remind the reader here? }

\subsection{Related Works}
\label{sec:related}

Rank aggregation from pairwise comparisons is an important problem in computer science, which has been widely studied under different comparison models. Most existing works focus on the non-active setting: the pairs of items for comparisons are fixed (or chosen completely at random) and the algorithm cannot adaptively choose the next pair for querying. In this non-active ranking setup, when the goal is to obtain a global ranking over all the items, Negahban et al. \cite{Negahban12RankCentrality} proposed the \emph{RankCentrality} algorithm under the popular Bradley-Terry model, which is a special case of the MNL model for pairwise comparisons. Lu and Boutilier \cite{Craig:11} proposed a ranking algorithm under the Mallows model. Rajkumar and Agarwal  \cite{Rajkumar14} investigated different statistical assumptions (e.g., generalized low-noise condition) for guaranteeing to recover the true ranking. Shah et al. \cite{Shah15Sto} studied the ranking aggregation under a non-parametric comparison model---strong stochastic transitivity (SST) model, and converted the ranking problem into a matrix estimation problem under shape-constraints. Most machine learning literature assumes that there is a true global ranking of items and the output of each pairwise comparison follows a probabilistic model. Another way of formulating the ranking problem is via the minimum feedback arc set problem on tournaments, which does not assume a true global ranking and aims to find a ranking that
minimizes the number of inconsistent pairs. There is a vast literature on the minimum feedback arc set problem and here we omit the survey of this direction (please see \cite{Mathieu07} and references therein).  %\xnote{Xi: Please check my description of feedback arc set?}
Due to the increasing number of items, it is practically more useful to identify the top-$k$ items in many internet applications. Chen and Suh \cite{Chen15}, Jang et al. \cite{Jang13},  and Suh et al. \cite{Suh16Adversarial} proposed various spectral methods for top-$k$ item identification under the BTL model or mixture of BTL models. Shah and Wainwright \cite{Shah15Sim} proposed a counting-based algorithm under the SST model and Chen et al. The notion of instance optimal was originally defined and emphasized as an important concept in \cite{FaginLN03} for identifying the top-$k$ objects from sorted lists. \cite{Chen:17:SODA} suggested that notion ``instance optimal'' is necessary for rank aggregation from noisy pairwise comparisons in complicated noise models and further improved \cite{Shah15Sim} by proposing an algorithm that has competitive ratio $\tilde{\Theta}(\sqrt{n})$ compared to the best algorithm of each instance and proving $\tilde{\Theta}(\sqrt{n})$ is tight.

In addition to static rank aggregation,  active noisy sorting and ranking problems have received a lot of attentions in recent years. For example, several works \cite{Braverman08, Ailon11, Jamieson11,Wauthier13} studied the active sorting problem from noisy pairwise comparisons and explored the sample complexity to approximately recover the true ranking in terms of some distance function (e.g., Kendall's tau). Chen et al. \cite{Chen13} proposed a Bayesian online ranking algorithm under the mixture of BTL models. Dwork et al. \cite{Dwork01} and Ailon et al. \cite{Ailon08} considered a related Kemeny optimization problem, where the goal is to determine the total ordering that minimizes the sum of the distances to different permutations. For top-$k$ identification, Braverman et al. \cite{BMW16} initiated the study of how round complexity of active algorithms can affect the sample complexity. Sz\"{o}r\'{e}nyi et al. \cite{Szorenyi:15} studied the case of $k=1$ under the BTL model.   Heckel et. al. \cite{Heckel:16} investigated the active ranking under a general class of nonparametric models and also established a lower bound on the number of comparisons for parametric models. A very recent work by Mohajer and Suh \cite{Mohajer:16:active} proposed an active algorithm for top-$k$ identification under a general class of pairwise comparison models, where the instance difficulty is characterized by the key quantity $\min_{i \in \{1,\ldots, k\}} \min_{j: j>i} (p_{ij}-0.5)^2$. Here, $p_{ij}$ is the probability of item $i$ is preferred over item $j$. However, according to our result in Theorem \ref{thm:lower}, the obtained sample complexities in previous works are not instance optimal. We note that the lower bound result in Theorem \ref{thm:lower} holds for algorithms even when all the values of $\theta_i$'s are known (but without the knowledge of which item corresponds to which value) and thus characterizes the difficulty of each instance.
Moreover, we study the the multi-wise comparisons, which has not been explored in ranking aggregation literature but has a wide range applications.

Finally, we note that the top-$k$ ranking problem is related to the best $k$ arm identification in multi-armed bandit literature \cite{Bubeck:13,Jamieson:14,Zhou:14,Chen:17:adaptive}. However, in the latter problem, the samples are \emph{i.i.d.} random variables rather than comparisons and the goal is to identify the top-$k$ distributions with largest means.

\section{Algorithm}
\label{sec:ub}

%!TEX root =  note.tex

For notational simplicity, throughout the paper we use the words w.h.p. to denote with probability $1 - 1/n^c$ for sufficiently large constant $c$.

\subsection{Top-$k$ item identification (For logarithmic $l$)}

For  $l  = O(\log n)$, we can always use pairwise comparisons by losing a polylog factor. Therefore, we only focus on the case when $l = 2$ in this section.

Before presenting the algorithm, let us first consider a graph $G = (V= [n], E)$ where each edge is labeled with either $\approx_l, \geq_l, \leq_l , >_l $ or $<_l$ (see Line \ref{algo:monotone} in Algorithm \ref{alg:malg_3}). Based on the labeling of edges, we give the following definition of label monotone, which will be used in Algorithm \ref{alg:malg_3}.
\begin{definition}[Monotone]
We call a path $i_1 \to i_2 \to \cdots \to i_d$ strictly label monotone if:
\begin{enumerate}
\item For every $j \in [d - 1]$, the edge $(i_j, i_{j + 1})$ is labeled with either $\approx_l, \geq_l$ or $>_l$. \label{p:1}
\item There exists at least one edge  $(i_j, i_{j + 1})$ with label $>_l$.
\end{enumerate}
Moreover, we call a path ``label monotone'' if only property \ref{p:1} holds.
\end{definition}

\begin{theorem}
\label{thm:ublogl}
For every $m$ items with $\theta_1 \geq \theta_2 \geq \cdots \geq \theta_n$, Algorithm  \ref{alg:malg_3}, on given a random permutation of labels $\Omega=[n]$ and $k$,  returns top-$k$ items w.h.p. using
$$ O \left(\kappa^7 \cdot  \left( k + \sum_{i = k + 1}^n \frac{\theta_{k }^2}{(\theta_k - \theta_i)^2}   +\sum_{i = 1}^k \frac{\theta_{k + 1 }^2}{(\theta_{ k + 1} - \theta_i)^2}  \right)\right) $$
total number of pairwise comparisons.
\end{theorem}

For the page limit, we defer the proof of Theorem \ref{thm:ublogl} to Appendix \ref{sec:ubapp}. We only provide the pseudocode in Algorithm \ref{alg:malg_3}. In Algorithm \ref{alg:malg_3}, we note that a different letter $m$ (instead of $n$) is used for denoting the set size because we will run the algorithm recursively with smaller sets. And also notice that the parameter $\kappa = \Omega( \log^2 n)$ regardless of the value of $m$. We also defer our result for superlogarithmic $l$ to Appendix \ref{sec:ubapp}.
%\xnote{Comments: (1) why change the size $n$ to $|\Omega|=m$ here? (2) Do we want to briefly explain $\Omega_b$ and $\Omega_g$? (3) Should we formally define the query results $Q_{u,q}$?}

\newcommand{\myalgthree}{\textsf{AlgPairwise}}
%\xnote{Maybe we should have a better name of Rank1, which sounds quite arbitrary}
\begin{algorithm}
\caption{\myalgthree}\label{alg:malg_3}
\begin{algorithmic}[1]
\State $\bold{Parameter}$: $\kappa = \Omega( \log^2 n)$.
\State $\bold{Input}$: A set of randomly permuted labels $\Omega$ with $|\Omega| = m$, $k$: number of top items. 
\State Uniformly at random sample $s = m \kappa $ subsets $S_1, \cdots S_s$ of $\Omega$, each of size $2$. Associate these subsets with a graph $G = (\Omega, E)$, where each edge $e_u \in E$ consists of all the vertices in $S_u$ for $u \in [s]$.
\State $q = 0$, $\Omega_{g} = \emptyset, \Omega_{b} = \emptyset, S = \emptyset$.
\While{true}
\State $q \leftarrow q + 1$.
\State Query each set $1$ time, obtain in total $s$ query results $\{R_{u, q} \}_{u \in [s]}$. ($R_{u,q}$ indicates the reported most favorable item)
\State For all $u \in [s]$, for $\{i, j\}= S_u$, let $\tilde{\theta}_{i, j} =  \frac{1}{q}\sum_{p \in [q]} 1_{R_{u, p}  = i} \enspace.$
\State For each edge $(i, j)\in E$,  we label it as: \label{algo:monotone}
\begin{enumerate}
\item $i \approx_l j$ if $ \frac{\tilde{\theta}_{i, j} }{\tilde{\theta}_{j, i} } \in \Bigl[ \frac{1}{  1 +  4\sqrt{ \frac{\kappa}{q}}} ,  1 + 4\sqrt{\frac{\kappa}{q}} \Bigr]$
\item $i \geq_l j$ if $ \frac{\tilde{\theta}_{i, j} }{\tilde{\theta}_{j, i} } \in \Bigl(1 + 4 \sqrt{\frac{\kappa}{q}} ,     1 + 32 \kappa  \sqrt{\frac{\kappa}{q}} \Bigr)$
\item  $i >_l  j$ if $\frac{\tilde{\theta}_{i, j} }{\tilde{\theta}_{j, i} }  \in \Bigl[ 1 + 32 \kappa  \sqrt{\frac{\kappa}{q}}, \infty \Bigr)$
\item $i \leq_l j$ if $\frac{\tilde{\theta}_{i, j} }{\tilde{\theta}_{j, i} }  \in \Bigl( \frac{1}{ 1 + 32 \kappa  \sqrt{\frac{\kappa}{q}}}  , \frac{1}{  1 +  4\sqrt{ \frac{\kappa}{q}}}  \Bigr)$
\item  $i <_l  j$ if $\frac{\tilde{\theta}_{i, j} }{\tilde{\theta}_{j, i} }  \in \Bigl[0,  \frac{1}{ 1 + 32 \kappa  \sqrt{\frac{\kappa}{q}}} \Bigr]$
\end{enumerate}
\State For every $i, j \in [m]$, we call $j \texttt{>>}_l i$ if there exists a strictly label monotone path of length at most $\kappa$ from $j$ to $i$.
\State For each $i \in [m]$, if there exists at least $k$ many $j \in [m]$ such that $j \texttt{>>}_l i$, then add $i$ to $\Omega_b$. ($\Omega_b$ is the subset of items that we are sure not in top-$k$.)
\State For each $i \in [m]$, if there exists at least $m - k$ many $j \in [m]$ such that $i \texttt{>>}_l j$, then add $i$ to $\Omega_g$. ($\Omega_g$ is the subset of items that we are sure in top-$k$.)
\State Break if $|\Omega_g \cup \Omega_b|  \geq \frac{m}{4} $.
\EndWhile
\State   $\Omega' = \Omega - \Omega_g - \Omega_b, k' = k- |\Omega_g|$, $S = S \cup \Omega_g \cup \myalgthree (\Omega', k')$.
\State $\bold{Return}$ $S$.
\end{algorithmic}
\end{algorithm}

\newpage

\section{Lower bounds}
\label{sec:lb}
We will prove lower bounds on the number of comparison used by any algorithm which identifies top-$k$ items even when the values of preference scores $\{\theta_1,...,\theta_n\}$ are given to the algorithm. (The algorithm just do not know which item has which $\theta_i$). For the page limit, all the proofs are deferred to Appendix \ref{app:lb}.

\subsection{Lower bounds for close weights}

\begin{theorem}
\label{thm:lb2ctk}
Assume $\theta_k > \theta_{k+1}$ and $c < 10^{-4}$. For any algorithm $A$ (can be adaptive), if $A$ uses $c \sum_{ i:k+1\leq i, \theta_i \geq \theta_k/2} \frac{\theta_k^2}{(\theta_k- \theta_i)^2}$ comparisons of any size (can be $l$-wise comparison for $2 \leq l \leq n$), then $A$ will identify the top-$k$ items with probability at most $7/8$.
\end{theorem}

\begin{theorem}
\label{thm:lb2}
Assume $\theta_k > \theta_{k+1}$ and $c < 4\cdot 10^{-4}$. For any algorithm $A$ (can be adaptive), if $A$ uses $c \sum_{ i: i\leq k, \theta_i \leq 2\theta_{k+1}} \frac{\theta_{k+1}^2}{(\theta_{k+1}- \theta_i)^2}$ comparisons of any size (can be $l$-wise comparison for $2 \leq l \leq n$), then $A$ will identify the top-$k$ items with probability at most $7/8$.
\end{theorem}

\subsection{Lower bounds for arbitrary weights}

\begin{theorem}
\label{thm:lb3}
Assume $c < 1/18$. For any algorithm $A$ (can be adaptive), if $A$ uses $c \sum_{ i: i> k} \frac{\theta_i}{\theta_k}$ comparisons of any size (can be $l$-wise comparison for $2 \leq l \leq n$), then $A$ will identify the top-$k$ items with probability at most $7/8$.
\end{theorem}

\begin{theorem}
\label{thm:lb4}
For any algorithm $A$ (can be adaptive), if $A$ uses $k/4$ comparisons of any size (can be $l$-wise comparison for $2 \leq l \leq n$), then $A$ will identify the top-$k$ items with probability at most $2/3$.
\end{theorem}

\begin{theorem}
\label{thm:lb5}
Assume $c < 1/2$. For any algorithm $A$ (can be adaptive), if $A$ uses $\frac{c n}{l} $ comparisons of size at most $l$ (can be $2$-wise, $3$-wise,...,$l$-wise comparisons), then $A$ will identify the top-$k$ items with probability at most $7/8$.
\end{theorem}

\subsection{Combining lower bounds}

\begin{corollary}[Restatement of Theorem \ref{thm:lower}]
For any algorithm $A$ (can be adaptive), suppose $A$ uses comparisons of size at most $l$ (can be $2$-wise, $3$-wise,...,$l$-wise comparisons). $A$ needs
\[
\Omega  \Biggl( \frac{n}{l}  + k + \frac{\sum_{i \geq k  + 1} \theta_{i}}{\theta_k} +  \sum_{i \geq k + 1, \theta_i \geq \frac{\theta_k}{2}} \frac{\theta_{k}^2}{(\theta_k - \theta_i)^2}   +\sum_{ i: i\leq k, \theta_i \leq 2\theta_{k+1}} \frac{\theta_{k + 1 }^2}{(\theta_{ k + 1} - \theta_i)^2}  \Biggr)
\]
to identify the top-$k$ items with probability at least $7/8$.
\end{corollary}

\begin{proof}
To prove this corollary, we just need to combine all the results in Theorem \ref{thm:lb2ctk}, Theorem \ref{thm:lb2}, Theorem \ref{thm:lb3}, Theorem \ref{thm:lb4} and Theorem \ref{thm:lb5}. And then use the fact that if $b < a_1 + \cdots + a_5$ then there exists $i \in \{1,2,3,4,5\}$ such that $b < 5a_i$.
\end{proof}

\appendix

%!TEX root =  note.tex

\section{Additional Results and Proofs of Section \ref{sec:ub}}
\label{sec:ubapp}

Throughout the proofs we are going to use the following claim which is a simple fact about the binomial concentration.

\begin{claim}[Binomial concentration] \label{claim:binomial_concentration}

For every $m \in \mathbb{N}^*$, every $p \in [0, 1]$, suppose $X \sim B(m, p)$, then $X \in [mp - O(\sqrt{mp\log n}) , mp + O(\sqrt{mp\log n})]$ w.h.p (with high probability respect to $n$).
\end{claim}

\subsection{Top-$k$ item identification (For logarithmic $l$)}

In this section, we prove Theorem \ref{thm:ublogl} of Section \ref{sec:ub}.

Following Claim \ref{claim:binomial_concentration}, we know that for every $(i, j) \in E$, every $q$, $\tilde{\theta}_{i, j}  \in \left[\theta_{i, j}  - \sqrt{\frac{\theta_{i, j} \kappa }{q} } ,  \theta_{i, j}  + \sqrt{\frac{\theta_{i, j} \kappa}{q} }\right]$ w.h.p. W.l.o.g, let us just focus on the case that this bound is satisfied for all $(i, j) \in E$ and every $q$.

We have the following Lemma about the labelling:
\begin{lemma}[Label] \label{lem:label}
For  $q = \Omega(\kappa^3) $, we have:
\begin{enumerate}
\item if $\theta_i \geq \theta_j$, then $i \approx_l j, i \geq_l j$ or $i >_l j$.
\item if $\theta_i \geq \theta_j \left(1  + 128 \kappa \sqrt{\frac{\kappa}{q} } \right)$, then $i >_l j$.
\item if $i \geq_l j $ or $ i \approx_l j$, then $${\theta_i} \geq  \theta_j \left( 1 - 8 \sqrt{\frac{\kappa}{q}} \right)$$
\item if $i >_l j$, then $${\theta_i} \geq \theta_j \left(1  + 16 \kappa \sqrt{\frac{\kappa}{q} } \right)$$
\end{enumerate}
\end{lemma}

\begin{proof}[Proof of Lemma \ref{lem:label}]

\begin{enumerate}
\item We know that for $q = \Omega(\kappa^3) $ and $\theta_i \geq \theta_j$:
\begin{align*}
\frac{\tilde{\theta}_{i, j} }{\tilde{\theta}_{j, i} }  &\geq \frac{\frac{\theta_i}{\theta_i + \theta_j} - \sqrt{\frac{\theta_i}{\theta_i + \theta_j}  \cdot \frac{\kappa}{q}}}{\frac{\theta_j}{\theta_i + \theta_j} +\sqrt{\frac{\theta_j}{\theta_i + \theta_j}  \cdot \frac{\kappa}{q}}} \geq \frac{\theta_i \left(1 - \sqrt{\frac{2 \kappa}{q}} \right)}{\theta_j + \sqrt{\frac{2  \theta_i \theta_j \kappa}{q}} }
\\
&\geq  \frac{ 1 - \sqrt{\frac{2 \kappa}{q}} }{\frac{\theta_j}{\theta_i} + \sqrt{\frac{\theta_j}{\theta_i} \cdot \frac{2   \kappa}{q}} } \geq  \frac{ 1 - \sqrt{\frac{2 \kappa}{q}} }{ 1 + \sqrt{\frac{2 \kappa}{q}} }
\\
&\geq  \frac{1}{  1 + 4\sqrt{ \frac{\kappa}{q}}}
\end{align*}

\item Again by $\theta_i \geq \theta_j \left(1  + 128 \kappa \sqrt{\frac{\kappa}{q} } \right)$ and $q = \Omega(\kappa^3) $, we know that $\frac{\theta_j}{\theta_i}  \leq 1 - 64 \kappa \sqrt{\frac{\kappa}{q}} $, therefore, we have:
\begin{align*}
\frac{\tilde{\theta}_{i, j} }{\tilde{\theta}_{j, i} }
&\geq  \frac{ 1 - \sqrt{\frac{2 \kappa}{q}} }{\frac{\theta_j}{\theta_i} + \sqrt{\frac{\theta_j}{\theta_i} \cdot \frac{2   \kappa}{q}} }  \geq \frac{ 1 - \sqrt{\frac{2 \kappa}{q}} }{\frac{\theta_j}{\theta_i} + \sqrt{ \frac{2   \kappa}{q}} }  \geq \frac{ 1 - \sqrt{\frac{2 \kappa}{q}} }{1 - 64 \kappa \sqrt{\frac{\kappa}{q}} + \sqrt{ \frac{2   \kappa}{q}} } \geq 1 + 32 \kappa \sqrt{\frac{\kappa}{q}} \enspace.
\end{align*}

\item Let us suppose $\theta_i \leq \theta_j$, otherwise we already complete the proof. Now, we have: $$ \frac{\frac{\theta_i}{\theta_i + \theta_j} + \sqrt{\frac{\theta_i}{\theta_i + \theta_j}  \cdot \frac{\kappa}{q}}}{\frac{\theta_j}{\theta_i + \theta_j} - \sqrt{\frac{\theta_j}{\theta_i + \theta_j}  \cdot \frac{\kappa}{q}}}  \geq \frac{\tilde{\theta}_{i, j} }{\tilde{\theta}_{j, i} } \geq  \frac{1}{  1 + 4\sqrt{ \frac{\kappa}{q}}} $$

Which implies that
$$\frac{\theta_i + \sqrt{\frac{2\theta_i \theta_j \kappa}{q} }}{\theta_j \left(1 - \sqrt{\frac{2 \kappa}{q}} \right)} \geq  \frac{1}{  1 + 4\sqrt{ \frac{\kappa}{q}}}$$

Therefore, by $\theta_i \leq \theta_j$, we have:
$$\frac{\frac{\theta_i}{\theta_j} + \sqrt{\frac{2 \kappa}{q} }}{ 1 - \sqrt{\frac{2 \kappa}{q}} } \geq  \frac{1}{  1 + 4\sqrt{ \frac{\kappa}{q}}}$$

Which implies that $$\frac{\theta_i}{\theta_j}  \geq   \frac{1 - \sqrt{\frac{2 \kappa}{q}} }{  1 + 4\sqrt{ \frac{\kappa}{q}}}  -  \sqrt{\frac{2 \kappa}{q} } \geq    1 - 8 \sqrt{\frac{\kappa}{q}} $$

\item  Let us suppose $\theta_i \leq 2\theta_j$, otherwise we already complete the proof. Again, we have:
$$\frac{\frac{\theta_i}{\theta_j} + \sqrt{\frac{3 \kappa}{q} }}{ 1 - \sqrt{\frac{3 \kappa}{q}} } \geq    1 + 32 \kappa\sqrt{ \frac{\kappa}{q}}$$

Which implies that $$\frac{\theta_i}{\theta_j}  \geq   \left(1 - \sqrt{\frac{3 \kappa}{q}} \right)\left(  1 + 32 \kappa \sqrt{ \frac{\kappa}{q}} \right) -  \sqrt{\frac{2 \kappa}{q} } \geq    1 + 16 \kappa \sqrt{\frac{\kappa}{q}} $$

\end{enumerate}

\end{proof}

Above, the Lemma \ref{lem:label} implies that w.h.p. the labelling of each edge $(i, j)$ is consistent with the order of $\theta_i, \theta_j$. Now, the algorithm will declare $i \texttt{>>}_l j $ if there exists strictly label monotone path from $i$ to $j$. Using the Lemma above we can show that if such path exists, then $\theta_i > \theta_j$. To show the other direction that such paths exists when $\theta_j > \theta_i$, we first consider the following graph Lemma that gives the exists of monotone path in random graph $G(m, p)$.
\begin{lemma}[Graph Path]\label{lem:graph_path}

For every $m \leq n$, every random graph $G(m, p)$ on vertices $V = [m]$, if $p \geq \frac{\kappa}{m}$, then w.h.p. For every $i, j \in [m]$ with $j \geq i + \frac{m}{4}$, there exists a path $i = i_ 1 \to i_2 \to \cdots i_d = j$ such that
\begin{enumerate}
\item $d \leq \kappa$.
\item $i_{r} \leq i_{r + 1}$ for every $r \in [d - 1]$.
\end{enumerate}
We call such a path a monotone path from $i$ to $j$.
\end{lemma}

\begin{proof}[Proof of Lemma \ref{lem:graph_path}]

It is sufficient to consider the case when $m = \Omega(\sqrt{\kappa})$, otherwise w.h.p. the graph is a complete graph and theorem is automatically true.

We consider a sequential way of generating $G(m, p)$: At each time $t= 1, 2, \cdots , m$, a vertex $t$ arrives and there exists an edge between $t$ and each $t' \in [t - 1]$ with probability $p$. Let us consider a fixed $i \leq \frac{3}{4}m$ and $j \geq i + \frac{m}{4}$. Let $\tau = \frac{\sqrt{\kappa}}{32} = \Omega(\log n)$. We will divide the set $\{i, i + 1, \cdots , j - 1\}$ into $\tau$ subsets $H_1, \cdots , H_\tau$ such that $$H_r = \left\{ i + (r - 1) \frac{j  - i}{\tau},  i + (r - 1) \frac{j  - i}{\tau} + 1, \cdots  , i + r \frac{j  - i}{\tau} - 1 \right\}  $$

Since $ m = \Omega(\sqrt{\kappa})$ we know that $| H_r | \geq 1$.

Let us define the random variable $Y_{r}, X_r, Z_v$ as:
$$Y_r = \text{the set of all $v \in H_r$ such that there exists a monotone path from $i$ to $v$ of length at most $r$} $$
and $X_r = |Y_r|$.

For each $v \in H_{r + 1}$, we define $$Z_v = 1_{\text{there is an edge between $v$ and at least one vertex in $Y_r$}}$$

Clearly, $X_1 \geq 1$ and each $Z_v$ is i.i.d. random variable in \{0, 1\} with $\Pr[Z_v = 0 \mid X_r] =  (1 - p)^{X_r} $. On the other hand, by definition,
$$X_{r +1} \geq \sum_{v \in H_{r + 1}} Z_v$$

We consider two cases:
\begin{enumerate}
\item $X_{r } \geq \frac{1}{p}$, then $\Pr[Z_v = 1]  \geq  \frac{1}{4}$.
\item $X_t < \frac{1}{p}$, then by $(1 - p)^x \leq 1 - \frac{xp}{2}$ for $x < 1/p$, we have $\Pr[Z_v = 1] \geq \frac{pX_r}{2}$.
\end{enumerate}

Consider a fixed $X_r$ and for each $v \in H_{r + 1}$, let $Z_v$ be the random variable. By standard Chernoff bound, we have:
\begin{enumerate}
\item If $X_{r } \geq \frac{1}{p}$, then w.h.p. $X_{r + 1} \geq \frac{j - i}{4 \tau} \geq \frac{m}{16 \tau}$.
\item $1 \leq X_r < \frac{1}{p}$, then w.h.p. $X_{r + 1} \geq \frac{(j - i)p X_r}{2 \tau} - \sqrt{ \tau \frac{(j - i)p X_r}{2 \tau}}    $.

Recall that $p \geq \frac{\kappa}{m}$ and $j - i \geq \frac{m}{4}$, therefore,
$$ \sqrt{ \tau \frac{(j - i)p X_r}{2 \tau}}    \leq   \frac{(j - i)p X_r}{4 \tau}, \quad  \frac{(j - i)p X_r}{4 \tau}  \geq 2 X_r$$

Which implies that w.h.p. $X_{r + 1} \geq  2 X_r$.
\end{enumerate}

Putting everything together, we know that for $\tau = \Omega(\log n)$, w.h.p. $X_{\tau} \geq \frac{m}{16 \tau}$. Therefore, condition on this event, by
\begin{align*}
&\Pr\left[\text{there is an edge between $j$ and $Y_{\tau}$} \bigg| X_{\tau} \geq \frac{m}{16 \tau} \right] = 1 -  (1 - p)^{ \frac{m}{16 \tau}}
\\
& \geq 1 -  \left(1 - \frac{\kappa}{m} \right)^{\frac{m}{16 \tau}} \geq  1 -  \left(1 - \frac{1024 \tau^2}{m} \right)^{\frac{m}{16 \tau}}  \geq 1 - \frac{1}{n^{\Omega(1)}}
\end{align*}

We complete the proof.

%and define the random variable $X_t (t \geq i)$  as
%$$X_t = \# \text{ of vertices $v$ ($v \leq t$) such that there exists a monotone path from $i$ to $v$}$$
%
%
%Clearly, $X_i = 1$ and $X_{t + 1} \in \{ X_t, X_t + 1\}$. We have the following formula that  associated  the expectation of $X_{t + 1}$ with the value of $X_t$:
%$$\E[X_{t + 1} - X_t \mid X_t] = 1 - (1 - p)^{X_t}$$
%
%
%We consider two cases:
%\begin{enumerate}
%\item $X_t \geq \frac{1}{p}$, then $\E[X_{t + 1} - X_t \mid X_t]  \geq  \frac{1}{2}$.
%\item $X_t < \frac{1}{p}$, then by $(1 - p)^x \leq 1 - \frac{xp}{2}$ for $x < 1/p$, we have: $\E[X_{t + 1} - X_t \mid X_t]  \geq \frac{p}{2}X_t$.
%\end{enumerate}
%
%We will show the following three claims:
%\begin{enumerate}
%\item If $X_{j - 1} \geq \frac{n}{16}$, then w.h.p. $X_{j} - X_{j - 1} = 1$.
%\end{enumerate}

\end{proof}

Having this Lemma, we can present the main Lemma above the algorithm:
\begin{lemma}[Main 3] \label{lem:main_3}
Suppose $q = \Omega(\kappa^3)$, then w.h.p. the following holds:
\begin{enumerate}
\item $\Omega_g \subseteq [k]$, $\Omega_b \cap [k] = \emptyset$.
\item If $k \leq \frac{m}{2 }$ and $q  = \Omega \left(\frac{\kappa^5}{m} \cdot \sum_{i = k + 1}^m \frac{\theta_{k }^2}{(\theta_k - \theta_i)^2} \right) $, then $|\Omega_b| \geq \frac{m}{4}$.
\item If $k> \frac{m}{2 }$ and $q  = \Omega \left( \frac{\kappa^5}{m}   \cdot \left( k  + \sum_{i = 1}^k \frac{\theta_{k + 1 }^2}{(\theta_{ k + 1} - \theta_i)^2} \right) \right) $, then $|\Omega_g| \geq \frac{m}{4}$.
\end{enumerate}
\end{lemma}

Since the algorithm terminates within $O(\kappa)$ recursions, moreover, in each recursion, the algorithm makes at most $\kappa m q$ queries. Therefore, Lemma \ref{lem:main_3} implies that the algorithm runs in total queries:
$$ O \left(\kappa^7 \cdot  \left( k + \sum_{i = k + 1}^m \frac{\theta_{k }^2}{(\theta_k - \theta_i)^2}   +\sum_{i = 1}^k \frac{\theta_{k + 1 }^2}{(\theta_{ k + 1} - \theta_i)^2}  \right)\right) $$

\begin{proof}[Proof of Lemma \ref{lem:main_3}]

\begin{enumerate}
\item It suffices to show that if $i \texttt{>>}_l j$, then $\theta_i \geq \theta_j$. To see this, consider a strictly label monotone path $i = i_1 \to i_2 \to \cdots \to i_d = j$ with length $d \leq \kappa$. By Lemma \ref{lem:label}, we know that for every $r \in [d - 1]$, we have:
$\theta_{i_r}\geq \theta_{i_{r + 1}}\left( 1 - 8 \sqrt\frac{\kappa}{q} \right)$. Moreover, there exists an $r' \in [d - 1]$ such that $\theta_{i_{r'}} \geq \theta_{i_{r' + 1}}\left( 1 + 16 \kappa\sqrt\frac{\kappa}{q} \right)$. Multiply every thing together, we know that
$$\theta_i \geq \theta_j \left( 1 - 8 \sqrt\frac{\kappa}{q} \right)^{\kappa - 1} \left( 1 + 16 \kappa \sqrt\frac{\kappa}{q} \right) \geq \theta_j$$

\item
Let us denote the set $H = \{ \frac{3}{4} m + 1, \frac{3}{4} m + 2, \cdots m \}$, we will prove that $H \subseteq \Omega_b$. Consider one $j \in H$, by Lemma \ref{lem:graph_path}, w.h.p. for every $i \in [k]$, there exists a path $i = i_ 1 \to i_2 \to \cdots i_d = j$ of length at most $\kappa$ such that $\theta_{i_r} \geq \theta_{i_{r + 1}}$ for every $r \in [d - 1]$. Now, by Lemma \ref{lem:label}, we know that this path is label monotone. We just need to show that this path is strictly label monotone. To see this, we know that there exists one $r' \in [d - 1]$ such that $$\theta_{i_{r'}} \geq \theta_{i_{r' + 1}}  \left(\frac{\theta_i}{\theta_j} \right)^{1/ \kappa}$$

Let $\nu = \frac{1}{m - k}\sum_{i = k + 1}^m \frac{\theta_{k }^2}{(\theta_k - \theta_i)^2}$. Now, since $k \leq \frac{m}{2}$, we can apply Markov inequality and conclude that for this $j \in H$ and $i \in [k]$, $\frac{\theta_{i}^2}{(\theta_{i} - \theta_j)^2} \leq \frac{\theta_{k}^2}{(\theta_{k} - \theta_j)^2} \leq 2\nu  $. Which implies that
$$\frac{\theta_i}{\theta_j} \geq \frac{1}{1 - \sqrt{\frac{1}{2 \nu}}} \geq 1 +  \sqrt{\frac{1}{2 \nu}}$$

For $q = \Omega(\kappa^5 \nu)$, we know that
$$\left(\frac{\theta_i}{\theta_j}  \right)^{1/\kappa} \geq \left(1 + 64 \kappa^2 \sqrt\frac{\kappa}{q} \right)^{1/\kappa} \geq \left( 1 + 32 \kappa\sqrt\frac{\kappa}{q} \right)$$

Therefore, $\theta_{i_{r'}} \geq \theta_{i_{r' + 1}}  \left(\frac{\theta_i}{\theta_j} \right)^{1/ \kappa} \geq  \left( 1 + 32 \kappa\sqrt\frac{\kappa}{q} \right)$. By definition, we shall label $i_{r'} > i_{r' + 1}$ and thus $i \texttt{>>}_l j$.

\item It can be shown with exactly the same calculation as 2 with $H = \left\{  1, 2, \cdots \frac{1}{4}m  \right\}$ and apply Markov inequality on
$$\nu = \frac{1}{k}\sum_{i = 1}^k \frac{\theta_{i }^2}{(\theta_{k + 1}- \theta_i)^2} =\frac{1}{k} O\left(k + \sum_{i = 1}^k \frac{\theta_{k + 1 }^2}{(\theta_{k + 1}- \theta_i)^2} \right) $$
\end{enumerate}

\end{proof}

\subsection{Top-$k$ item identification (For super logarithmic  $l$)}

Before presenting the algorithm, we first argue about which case using bigger $l$ is unnecessary. We have the following Claim:
\begin{claim}[Bigger $l$]\label{claim:big_l} For every $l \leq m \leq n$, we have:
\begin{align*}
& k + \sum_{i = k + 1}^m \frac{\theta_{k }^2}{(\theta_k - \theta_i)^2}   +\sum_{i = 1}^k \frac{\theta_{k + 1 }^2}{(\theta_{ k + 1} - \theta_i)^2}
\\
& \leq \left(\frac{m}{l}  + k + \frac{\sum_{i \geq k  + 1} \theta_{i}}{\theta_k} +  \sum_{i \geq k + 1, \theta_i \geq \frac{\theta_k}{2}} \frac{\theta_{k}^2}{(\theta_k - \theta_i)^2}   +\sum_{i = 1}^k \frac{\theta_{k + 1 }^2}{(\theta_{ k + 1} - \theta_i)^2}  \right) + 4m
\end{align*}
\end{claim}

Therefore, as long as we can show one of the following:
\begin{enumerate}
\item $k = \Omega(m)$.
\item $ \frac{\sum_{i \geq k  + 1} \theta_{i}}{\theta_k}  = \Omega(m)$.
\end{enumerate}

We can just use the algorithm for $l = 2$. Otherwise, we shall consider larger $l$, we will directly considering the case when $l= \Omega(\log n)$. Before giving the algorithm, it is convenient to first consider the following query procedure: For a fixed $Q$, do:

\newcommand{\myalgott}{\textsf{BasicQuery}}
\begin{algorithm}
\caption{\myalgott}\label{alg:malg_123}
\begin{algorithmic}[1]
\State $\bold{Parameter}$: $\kappa = \Omega(\log^2 n)$
\State $\bold{Input}$: $\Omega$: set of items with $|\Omega| = m$,  $k$: number of top items to find. $l$: size of the subset to query.
\State  Uniformly at random sample $s = \frac{m \kappa}{l} $ subsets $S_1, \cdots S_s$ of $\Omega$, each of size $l$. Associate these subsets with a hypergraph $G = (\Omega, E)$, where each edge $e_u \in E$ is consists of all the vertices in $S_u$ for $u \in [s]$.
\State Query each set $Q$ time, obtain in total $s Q$ query results $\{R_{u, q} \}_{u \in [s], q \in [Q]}$.
 \end{algorithmic}
\end{algorithm}

\newpage

For a fixed $q \leq Q$,  let us consider a random variable $\tilde{\theta}_{i, S_u} \in [0, 1]$ defined as
$$\tilde{\theta}_{i, S_u} =  \frac{1}{q}\sum_{r \in [q]} 1_{R_{u, r}  = i}$$

 For each $i, u$ such that $i \in S_u$, let us define 0-1 valued function
$1_{i, u, \alpha, \beta, \gamma} $ such that $1_{i, u, \alpha, \beta, \gamma} = 1 $ if and only if all the following conditions hold:
\begin{enumerate}
\item $\tilde{\theta}_{i, S_u} \geq \frac{\alpha }{q}$.
\item There exists at least $\gamma l$ many of the $j \in S_u$ such that $\tilde{\theta}_{j, S_u}  \leq \beta \tilde{\theta}_{i, S_u}$.
\end{enumerate}

We also consider the random variable $X_{i, u, \alpha, \beta, \gamma} $ associated with this function, where the randomness is taken over the uniformly at random choice of $S_u$ conditional on $i \in S_u$, and the randomness of the outcome of the queries.

We prove the following main Lemma:
\begin{lemma}[Indicator]\label{lem:indicator}

Let $\gamma \in \left[\frac{1}{32}, \frac{1}{2} \right], \beta \in (0, 32]$, $\alpha = \Omega(\kappa)$. For every $i \in [m]$, the following holds:

\begin{enumerate}
\item If  $q = \Omega\left(\alpha + \frac{2 \alpha l \sum_{j \in [m]} \theta_j} {m \theta_i }  \right) $ and $\theta_i \geq 2 \beta \theta_{(1 - 2 \gamma)m}$, then
$$\Pr[X_{i, u, \alpha, \beta, \gamma}  = 1] \geq \frac{15}{16}$$
\item For every $q$, if $\theta_i \leq  \frac{\beta}{2} \theta_{\left(1 - \gamma \right)m}$, then
$$\Pr[X_{i, u, \alpha, \beta, \gamma}  = 1] \leq \frac{9}{16}$$
\end{enumerate}
\end{lemma}

\begin{proof}[Proof of Lemma \ref{lem:indicator}]

\begin{enumerate}
\item

We first bound the probability that $\tilde{\theta}_{i, S_u} \geq \frac{\alpha }{q}$. By $$\tilde{\theta}_{i, S_u} \in \left[{\theta}_{i, S_u} - \sqrt{\frac{{\theta}_{i, S_u} \kappa}{q}} , {\theta}_{i, S_u} + \sqrt{\frac{{\theta}_{i, S_u} \kappa}{q}} \right]$$

we know that
$$\theta_{i, S_u}  \geq \frac{2 \alpha}{q} \implies \tilde{\theta}_{i, S_u } \geq \frac{\alpha}{q}$$

To lower bound this probability, we just need to consider the probability that $\theta_{i, S_u}  < \frac{2 \alpha }{q}$.  We apply Markov inequality and have that:
$$\Pr\left[  \theta_{i, S_u}  < \frac{2 \alpha }{q} \right] = \Pr\left[   \frac{2 \alpha }{q \theta_{i, S_u}} > 1 \right]   < \frac{\E \left[ \frac{2 \alpha }{q  \theta_{i, S_u} } \right ]}{1}$$

Notice that
$$\frac{2 \alpha}{q \theta_{i, S_u}} = \frac{2 \alpha \sum_{j \in S_u} \theta_j}{\theta_i q}$$

Therefore, $$\E\left[\frac{2 \alpha}{q \theta_{i, S_u}} \right] =  \frac{2 \alpha \E[\sum_{j \in S_u}] \theta_j}{\theta_i q}  \leq \frac{2 \alpha}{q}  + \frac{2 \alpha l \sum_{j \in [m]} \theta_j} {m \theta_i q} \leq \frac{1}{64}$$

Putting together we obtain $$\Pr\left[\tilde{\theta}_{i, S_u } < \frac{\alpha}{q} \right] \leq \Pr\left[{\theta}_{i, S_u } < \frac{2\alpha}{q} \right] < \frac{1}{64}$$

Now we can move to the second condition. For now, suppose $\tilde{\theta}_{i, S_u} \geq \frac{\alpha }{q}$ holds, we then know that
$$\tilde{\theta}_{i, S_u} \in \left[ \frac{31}{32}   {\theta}_{i, S_u},  \frac{33}{32}{\theta}_{i, S_u}\right]$$

Therefore, $\theta_i \geq 2 \beta \theta_{(1 - 2 \gamma)m}$ implies that for every $j \in H = \{ (1 - 2 \gamma)m, (1 - 2 \gamma)m + 1 , \cdots m\}$ with $j \in S_u$, we have:
\begin{align*}
\tilde{\theta}_{j , S_u} &\leq  \theta_{j, S_u} + \sqrt{\frac{\theta_{j, S_u} \kappa}{q}}   \leq  \theta_{j, S_u} + \sqrt{\frac{\theta_{i, S_u} \kappa}{q}}
\\
& \leq \theta_{j, S_u} + \frac{1}{128}\theta_{i, S_u} \leq \frac{ \theta_{i, S_u}  }{2 \beta} +  \frac{1}{128}\theta_{i, S_u}
\\
& \leq \frac{ 3\theta_{i, S_u}  }{ 4\beta} \quad \text{for $\beta \leq 32$}
\\
& \leq \frac{\tilde{\theta}_{i , S_u}}{\beta}
\end{align*}

Since $|H| = 2 \gamma m$, we know that for $l = \Omega(\log n)$, $ \Pr[ |H \cap S_u| < \gamma l ]  \leq \frac{1}{64}$.
Therefore,
$$\Pr[X_{i, u, \alpha, \beta, \gamma}  = 1] \geq 1 - \Pr\left[\tilde{\theta}_{i, S_u } < \frac{\alpha}{q} \right] - \Pr[ |H \cap S_u| < \gamma l ]  \geq \frac{15}{16}$$

\item The proof follows from the same calculation. Notice that this time we already have $X_{i, u, \alpha, \beta, \gamma}  = 1 \implies \tilde{\theta}_{i, S_u} \geq \frac{\alpha }{q} $.
\end{enumerate}

\end{proof}

For fixed $\alpha = \Omega(\kappa)$, every $\beta \in (0, 32], \gamma \in \left[\frac{1}{32}, \frac{1}{2} \right]$ and every $\tau  \in \left[ \frac{3}{4}, \frac{7}{8} \right]$, we consider set $$\Omega_{\beta, \gamma, \tau} = \left\{i \in [m] \bigg| \sum_{u: u \in [s],  i \in S_u} X_{i, u, \alpha, \beta, \gamma} \geq \tau \text{deg}(i)\right\}$$

We also have the following Corollary of Lemma \ref{lem:indicator}:

\begin{corollary}\label{cor:main:4}
\begin{enumerate}

\item For every $i, j \in [m]$ with $\theta_i \geq \theta_j$, every $\tau \in \left[ \frac{3}{4} \times \frac{33}{32}, \frac{7}{8} \right]$, w.h.p. $j \in \Omega_{\beta, \gamma, \tau}  \implies i \in  \Omega_{\beta, \gamma,  \frac{32}{33}\tau} $.
\item For every $i \in [m]$, if $\theta_{i} \geq 2 \beta \theta_{(1- 2 \gamma)m}$ and $q = \Omega\left(\alpha + \frac{2 \alpha l \sum_{j \in [m]} \theta_j} {m \theta_i }  \right)  $, then w.h.p. $i \in \Omega_{\beta, \gamma, \tau} $.

\item For every $q$, if $i \in \Omega_{\beta, \gamma, \tau}$, then w.h.p. $\theta_i \geq \frac{\beta}{2} \theta_{(1 - \gamma) m}$.
\end{enumerate}
\end{corollary}

Having this Corollary, we can do the following algorithm that selects all the $\theta_i \geq 32 \max \{\theta_{k}, \theta_{\frac{3}{4} m} \}$ and removes most of the $\theta_j \leq \frac{1}{4} \theta_k$:

\newcommand{\myalgfour}{\text{AlgMulti-wise}}
\begin{algorithm}
\caption{\myalgfour}\label{alg:malg_4}
\begin{algorithmic}[1]
\State $\bold{Parameter}$: $\kappa = O(\log^2 n)$.
\State $\bold{Input}$: $\Omega$: set of items,  $k$: number of top items to find.
\State $\bold{Output}$: $S$: set of top items. $\Omega'$: set of remaining items.
\State $\bold{Initialization}$: $S = \emptyset, \Omega' = \Omega$, $m = |\Omega|$.
\State Call $\myalgott$ to obtain $\{ \tilde{\theta}_{i, S_u} \}_{i \in [n], u \in [s]}$.
\If{$k \leq \frac{1}{2} m$}
\If{ $ 1 \leq |\Omega_{32, \frac{1}{4}, \frac{13}{16}} |$ and  $ |\Omega_{4, \frac{1}{16}, \frac{13}{16}} |< k$ }
%\State $R = 1$.
\State  $S_1 = \Omega_{4, \frac{1}{16}, \frac{7}{8}}$, $\Omega'' =\Omega - S_1$,  $(S', \Omega')  = \myalgfour(\Omega'', k - |S_1|, R)$, $S = S \cup S_1 \cup S''$.
\State Notice that we pick those numbers so $\frac{7}{8} \geq \frac{33}{32} \cdot \frac{13}{16} \geq \left( \frac{33}{32} \right)^2\cdot\frac{3}{4}$.
\ElsIf{$|\Omega_{4, \frac{1}{16}, \frac{13}{16}} |\geq k$}
\State  $\Omega'' =  \Omega_{4, \frac{1}{16}, \frac{3}{4}} $, $(S', \Omega')  = \myalgfour(\Omega'' , k, R)$, $S = S \cup S'$.

\EndIf

%\If{ $\Omega'  = \Omega$}
%\State $\bold{Break}$.
%\EndIf

%\If{$k \leq 7/8 n$}
%\State If $|\Omega_{\frac{1}{2}, \frac{1}{16}} | \geq \frac{n}{32}$, then recurse on $\Omega' = \Omega_{\frac{1}{2}, \frac{1}{16}}, k' = k$. \label{line:alg_4_1}
%\State If $|\Omega_{8, \frac{1}{4}} | \geq 1$, output $S = \Omega_{8, \frac{1}{4}}$, recurse on $\Omega' = \Omega -S, k'  = k  - |S |$.
%\Else
%\If{$|\Omega_{2, \frac{1}{8}, \frac{3}{4}} | \geq 1$}
%\State $S_1 = \Omega_{2, \frac{1}{8}, \frac{3}{4}} $,  $(S', \Omega') =  \myalgfour(\Omega - S_1, k - |S_1|)$, $S = S \cup S_1 \cup S'$.
%\EndIf
\EndIf

%
%\If{$|\Omega_{\frac{1}{2}, \frac{1}{16}, \frac{3}{4}}| \geq \frac{1}{8} n$}
%\State   $(S', \Omega') =  \myalgfour(\Omega' \backslash \Omega_{\frac{1}{2}, \frac{1}{16}, \frac{3}{4}} , k  - |S|)$,  $S = S \cup S'$.
%\EndIf

\State $\bold{Return}$ $S, \Omega'$.
 \end{algorithmic}
\end{algorithm}

We have the following lemma.
\begin{lemma}
For every $m$, every $\theta_1 \geq \theta_2 \geq \cdots \geq \theta_m$, every $k \leq m$, every $l \leq m$, Algorithm \ref{alg:malg_4}, on given a random permutation of $\Omega = [m], k$ satisfies:
\begin{enumerate}
\item Output set $(S, \Omega')$ of the algorithm satisfies $S \subseteq [k]$.
\item 
If $Q =  \tilde{\Omega}\left(1 + \frac{ l\left(k + \sum_{j \geq k} \theta_j \right)} {m \theta_k }  \right) $, then the algorithm returns in $O(\log m)$ many recursion calls, and 
after the algorithm, let us for simplicity still denote $[|\Omega'|] = \Omega'$ with $\theta_1 \geq \theta_2 \geq \cdots \theta_{|\Omega'|}$ and $k' = k - |S|$, then either
\begin{enumerate}
\item For every $i \in \Omega'$, $\theta_i \leq 512   \theta_{\frac{7}{8} |\Omega'|}$.
\item Or $ k' \geq \frac{1}{2}  |\Omega'|$.
\end{enumerate}
\end{enumerate}
\end{lemma}

\begin{proof}[Proof of the main theorem]
After running this algorithm, we can simply apply the algorithm for $l = 2$ (By Claim \ref{claim:big_l}), since one of the following is true:
\begin{enumerate}
\item $ \frac{\sum_{i \geq k'  + 1} \theta_{i}}{\theta_{k'}}  \geq \frac{3}{8}\times  \frac{1}{512} |\Omega'|$.
\item $ k' > \frac{1}{2} |\Omega'|$.
\end{enumerate}
Therefore, putting everything together, we can get the top $k$ items in total number of queries:
$$ \tilde{O}\left(\frac{m}{l}  + k + \frac{\sum_{i \geq k  + 1} \theta_{i}}{\theta_k} +  \sum_{i \geq k + 1, \theta_i \geq \frac{\theta_k}{2}} \frac{\theta_{k}^2}{(\theta_k - \theta_i)^2}   +\sum_{i = 1}^k \frac{\theta_{k + 1 }^2}{(\theta_{ k + 1} - \theta_i)^2}  \right) $$

\end{proof}

%\begin{enumerate}
%\end{enumerate}

Now, it just remains to prove this Lemma:

\begin{proof}

We first prove the correctness: $S \subseteq [k]$. We have the following observations:

\begin{enumerate}
\item If $|\Omega_{4, \frac{1}{16}, \frac{13}{16}} | \geq k$, then there must be $i \in \Omega_{4, \frac{1}{16}, \frac{13}{16}} $ with  $\theta_i \leq \theta_k$. By Corollary \ref{cor:main:4}, since $\frac{13}{16} \geq \frac{33}{32} \cdot \frac{3}{4}$, we know that $[k] \subseteq  \Omega_{4, \frac{1}{16}, \frac{3}{4}} $.
\item If there exists $j \in  \Omega_{4, \frac{1}{16}, \frac{7}{8}}$ such that $j \notin [k]$, then by $\theta_j \leq \theta_k$, apply  Corollary \ref{cor:main:4} with $\frac{7}{8} \geq \frac{33}{32} \cdot \frac{13}{16}$, we know that $[k] \subseteq  \Omega_{4, \frac{1}{16}, \frac{13}{16}} $, which implies $|\Omega_{4, \frac{1}{16}, \frac{13}{16}} | \geq k$. Therefore, we will not include any item that is not top $k$ to $S$ when recursing from Line 7.
\end{enumerate}

These two observations immediately imply $S \subseteq [k]$.

Now, we will show that for sufficiently large $Q$,  either of the two conditions hold:
\begin{enumerate}
\item For every $i \in \Omega'$, $\theta_i \leq 512   \theta_{\frac{7}{8} |\Omega'|}$.
\item Or $ k' \geq \frac{1}{2}  |\Omega'|$.
\end{enumerate}

 Let us for notation simplicity  drop the $'$ here. Clearly, we just need to consider the case when $k < \frac{1}{2} m$, otherwise, the algorithm will just terminate and the second condition is true.  We will first prove that $\theta_i \leq 64 \theta_{k}$ and then we prove that $\theta_k \leq 8 \theta_{\frac{7}{8} m} $.
\begin{enumerate}
% Which implies that there must be one previous recursion such that $|\Omega_{4, \frac{1}{16}, \frac{13}{16}} | < k, R = 1$.
\item To prove $\theta_i \leq 64 \theta_{k}$, we suppose on the contrary that $\theta_1 > 64 \theta_k$. Apply Corollary \ref{cor:main:4} with $q = Q = \Omega\left(\alpha  + \frac{2 \alpha l \left( \sum_{j \geq 1} \theta_j \right)} {m \theta_1}  \right)   $, we have that  $1 \in \Omega_{32, \frac{1}{16}, \frac{3}{4}} $, which implies that  $|\Omega_{32, \frac{1}{4}, \frac{13}{16}} | > 0$, so the algorithm won't terminate, contradict.

\item Now, we need to show that $\theta_k \leq 8 \theta_{\frac{7}{8} n} $. We also on the contrary suppose that $\theta_k > 8 \theta_{\frac{7}{8} m}$. Since the algorithm termniates, by the previous claim, we know that in the last recursion, it must be the case that  $\theta_1 \leq 64 \theta_k$. Which implies that
$$Q = \Omega\left(\alpha  + \frac{2 \alpha l \left( k \theta_k + \sum_{j \geq k } \theta_j \right)} {m \theta_k}  \right)  =   \Omega\left(\alpha  + \frac{2 \alpha l \left(  \sum_{j \geq 1 } \theta_j \right)} {m \theta_k}  \right)  $$

Therefore, if $\theta_k > 8 \theta_{\frac{7}{8} m}$, then  by Corollary \ref{cor:main:4}  we know that  $[k] \subseteq  \Omega_{4, \frac{1}{16}, \frac{3}{4}} $, so the algorithm won't terminate.
\end{enumerate}

Finally,  we consider about the total number of recursions. Clearly, if the algorithm recurses through the second case, then $|\Omega'| \leq \frac{15}{16} | \Omega|$. If the algorithm recurses through the first case, then  by Corollary \ref{cor:main:4}, it must be the case that
$$\theta_1 \geq 16 \theta_{\frac{7}{8} m}$$

Which implies that for all $i$ with $\theta_i \geq \frac{\theta_1}{2} \geq 8 \theta_{\frac{7}{8} m}$, $i \in  \Omega_{4, \frac{1}{16}, \frac{7}{8}}$.

Therefore, the total number of recursions of the algorithm is bounded by $O(\log m)$. So the total number of queries of the algorithm is:
$$O \left(\alpha  + \frac{2 \alpha l \left( k \theta_k + \sum_{j \geq k} \theta_j \right)} {m\theta_k}  \right)  \times \frac{\kappa m}{l}  \times O(\log m) = \tilde{O}\left(\frac{m}{l} + k + \frac{\sum_{j \geq k} \theta_j}{\theta_k} \right)$$

\end{proof}

$\bold{Remark}$: How to obtain the value $Q$: In the proof above we assumed that we have an aprior estimation of the value of $Q$. We can replace this assumption by initially setting $Q$ to be $Q = Q_0 = 1$, and run algorithm \ref{alg:malg_4} with $Q_0$ queries and then run the algorithm with pairwise comparision. Once the later algorithm requires more than $Q_0 \times \frac{ n}{l}$ queries, then we stop it, set $Q_1 = 2 Q_0$ and repeat this procedure with $Q = Q_1$. We keep on repeating this for $Q_2 = 2 Q_1, Q_3 = 2 Q_2, \cdots$ until the later algorithm requires less than $Q_i \times \frac{ n}{l}$ queries.

By the Lemma we just proved, the output of the algorithm is correct for every $Q$. Moreover, if $$Q \times \frac{n}{l} =  \tilde{\Omega}\left(\frac{n}{l}  + k + \frac{\sum_{i \geq k  + 1} \theta_{i}}{\theta_k} +  \sum_{i \geq k + 1, \theta_i \geq \frac{\theta_k}{2}} \frac{\theta_{k}^2}{(\theta_k - \theta_i)^2}   +\sum_{i = 1}^k \frac{\theta_{k + 1 }^2}{(\theta_{ k + 1} - \theta_i)^2}  \right) $$

Then this process will terminate, and the total query complexity is then bounded by  $\tilde{O}\left( Q \times \frac{n}{l} \right)$.

\section{Proofs of Section \ref{sec:lb}}
\label{app:lb}
\subsection{Lower bounds for close weights}

%!TEX root =  note.tex

\begin{theorem}[Restatement of Theorem \ref{thm:lb2ctk}]
\label{thm:lb2ctkapp}
Assume $\theta_k > \theta_{k+1}$ and $c < 10^{-4}$. For any algorithm $A$ (can be adaptive), if $A$ uses $c \sum_{ i:k+1\leq i, \theta_i \geq \theta_k/2} \frac{\theta_k^2}{(\theta_k- \theta_i)^2}$ comparisons of any size (can be $l$-wise comparison for $2 \leq l \leq n$), then $A$ will identify the top-$k$ items with probability at most $7/8$.
\end{theorem}

\begin{proof}
For notation convenience, we set $w_i = \frac{\theta_k^2}{(\theta_k- \theta_i)^2}$ for $i$ such that $i \geq k+1$ and $\theta_i \geq \theta_k / 2$. For other $i$, we set $w_i = 0$. We also set $W = \sum_{i=1}^n w_i$. Then we have $T = c W$.

First of all, we can assume $A$ is deterministic. This is because if $A$ is randomized, we can fix the randomness string which makes $A$ achieves the highest successful probability.

Let $S = (S_1,...,S_T)$ be the history of algorithm. Each $S_t$ is the comparison result of round $t$. Notice that since $A$ is deterministic, with $S_1,...,S_t$, we can determine the labels of items $A$ want to compare in round $t+1$ even when $A$ is adaptive. So there is no point to put the labels of compared items in the history. So we only put the comparison result in the history, i.e $S_t$ is a number in $[n]$ and $S$ is a length-$T$ string of numbers in $[n]$.

Again since $A$ is deterministic, the label $A$ outputs is just a deterministic function of $S$, we use $A(S)$ to denote it. $A$ outputs correctly if $A$ outputs the label of the top-$k$ items, i.e. $A(S) = \{ \pi_1,...,\pi_k\}$.

We use $p(S, \pi)$ to denote the probability that the items are labeled as $\pi$ and $A$ has history $S$. Now consider the case when we set $\theta_i$ equals to $\theta_k$ for $i \geq k+1$. In this case the probability of $A(S) =\{ \pi_1,...,\pi_k\}$ should be at most $1/2$ as item $k$ and item $i$ have the same weight. We use $p_i(S, \pi)$ to denote the probability that the items are labeled as $\pi$ and $A$ has history $S$ when $\theta_i$ is changed to $\theta_k$.

Now we prove the following lemma that gives the connection between $p(S,\pi)$ and $p_i(S,\pi)$.

\begin{lemma}
\label{lem:bounddiv}
Consider $p$ as a distribution over $(\pi,S)$. For all $c_1 > 0$, we have
\[
\Pr_{(\pi,S) \sim p} \left[\left( \sum_{i=1}^n \frac{w_i}{W} \ln \frac{p_i(S,\pi)}{p(S,\pi)} \right) \leq -c_1 \right] \leq \exp\left( -\frac{(c_1-4c)^2}{72c}\right).
\]

\end{lemma}
\begin{proof}
Define random variable $Z_t$ to be the following for $t = 1,...,T$ when $(\pi,S)$ is sampled from distribution $p$:
\[
Z_t  = \sum_{i=1}^n \frac{w_i}{W} \ln \frac{p_i(S_1...S_t,\pi)}{p(S_1...S_t,\pi)}.
\]
We have
\[
Z_T = \sum_{i=1}^n \frac{w_i}{W} \ln \frac{p_i(S,\pi)}{p(S,\pi)}.
\]
Now we want to show that sequence $0, Z_1+\frac{4}{W}, ..., Z_t +\frac{4t}{W}, ...,Z_T +  \frac{4T}{W}$ forms a supermartingale.

Suppose in round $t$, given $S_1,...,S_{t-1}$ and $\pi$, Algorithm $A$ compares items in set $U_t$. Let $\theta_{-i} = \sum_{j \in U_t, j\neq i} \theta_i$. Then we have, with probability $\theta_i / (\theta_i +\theta_{-i})$,
\[
Z_t - Z_{t-1} =  \frac{w_i}{W} \ln\left(  1+ \frac{(\theta_k- \theta_i)\theta_{-i}}{(\theta_k+\theta_{-i})\theta_i}\right) +  \sum_{j \in U_t, j \neq i} \frac{w_j}{W} \ln \left( 1- \frac{\theta_k - \theta_j}{\theta_k + \theta_{-j}} \right)
\]
Here are two simple facts about $\ln$. For $0 \leq x \leq 1$, $\ln(1+x) \geq x - x^2$. For $0 \leq x \leq 1/2$, $ \ln(1-x) \geq -x - x^2$. It's easy to check that for $i$ such that $w_i > 0$, we have $\frac{(\theta_k- \theta_i)\theta_{-i}}{(\theta_k+\theta_{-i})\theta_i} \leq  1$ and $\frac{\theta_k - \theta_i}{\theta_k + \theta_{-i}} \leq 1/2$ . Therefore, by these two facts, for $i$ such that $w_i  >0$, we have
\begin{eqnarray*}
&& \frac{\theta_i}{\theta_i + \theta_{-i}} w_i \ln\left(  1+ \frac{(\theta_k- \theta_i)\theta_{-i}}{(\theta_k+\theta_{-i})\theta_i}\right)  +\frac{\theta_{-i}}{\theta_i + \theta_{-i}} w_i \ln\left( 1- \frac{\theta_k - \theta_i}{\theta_k + \theta_{-i}} \right) \\
&\geq& \frac{ w_i}{\theta_i + \theta_{-i}}\left(  \frac{(\theta_k- \theta_i)\theta_{-i}}{\theta_k+\theta_{-i}} - \frac{(\theta_k- \theta_i)^2\theta_{-i}^2}{(\theta_k+\theta_{-i})^2\theta_i}-\frac{(\theta_k- \theta_i)\theta_{-i}}{\theta_k+\theta_{-i}} -\frac{(\theta_k - \theta_i)^2\theta_{-i}}{(\theta_k + \theta_{-i})^2}\right) \\
&=& -w_i \frac{(\theta_k - \theta_i)^2\theta_{-i}}{(\theta_k + \theta_{-i})^2 \theta_i}  = - \frac{\theta_{-i}\theta_k^2}{(\theta_k + \theta_{-i})^2 \theta_i} \\
&\geq&  - \frac{2\theta_{-i}\theta_k}{(\theta_k + \theta_{-i})^2 }\geq -\frac{4\theta_i}{\theta_k + \theta_{-i} } \geq-\frac{4\theta_i}{\theta_i + \theta_{-i} } .\\
\end{eqnarray*}

Therefore we have for all $t$ and $S_1,...,S_{t-1}$,
\[
\E[Z_t - Z_{t-1} | S_1,...,S_{t-1}] \geq  -\sum_{i \in U_t} \frac{4\theta_i}{W(\theta_i + \theta_{-i})}  \geq -\frac{4}{W}.
\]
As $Z_1,...,Z_{t-1}$ can be determined by $S_1,...,S_{t-1}$, we have for all $t$ and  $Z_1,...,Z_{t-1}$,
\[
\E\left[\left(Z_t  +\frac{4t}{W} \right)- \left(Z_{t-1}+\frac{4(t-1)}{W}\right) | Z_1 -\frac{4}{W},...,Z_{t-1} -\frac{4(t-1)}{W}\right] \geq  0.
\]
Therefore sequence $0, Z_1  +\frac{4}{W}, ..., Z_t +\frac{4t}{W}, ...,Z_T + \frac{4T}{W}$ forms a supermartingale.

Now we want to bound $|Z_t - Z_{t-1}|$. We know that for $0 \leq x \leq 1$, $|\ln (1+x)| \leq x$ and for $0 \leq x \leq 1/2$, $|\ln(1-x)| \leq 2|x|$. Therefore for $i$ such that $w_i > 0$,
\[
|\frac{w_i}{W} \ln\left(  1+ \frac{(\theta_k- \theta_i)\theta_{-i}}{(\theta_k+\theta_{-i})\theta_i}\right)| \leq \frac{w_i}{W} \cdot  \frac{(\theta_k- \theta_i)\theta_{-i}}{(\theta_k+\theta_{-i})\theta_i} \leq \frac{\theta_k^2\theta_{-i}}{W\theta_i(\theta_k -\theta_i)(\theta_k+\theta_{-i})}\leq \frac{2\theta_k}{W(\theta_k -\theta_{k+1})}
\]
and
\[
|\frac{w_i}{W} \ln \left( 1- \frac{\theta_k- \theta_i}{\theta_k+\theta_{-i}}\right)| \leq\frac{w_i}{W} \cdot \frac{2(\theta_k- \theta_i)}{\theta_k+\theta_{-i}} = \frac{2\theta_k^2}{W(\theta_k -\theta_i)(\theta_k+\theta_{-i})} \leq \frac{4\theta_k}{W(\theta_k -\theta_{k+1})} \cdot \frac{\theta_i}{\theta_i + \theta_{-i}}.
\]
Therefore, we get
\[
|Z_t -Z_{t-1}| \leq\frac{2\theta_k}{W(\theta_k -\theta_{k+1})} + \sum_{i \in U_t}\frac{4\theta_k}{W(\theta_k -\theta_{k+1})} \cdot \frac{\theta_i}{\theta_i + \theta_{-i}} \leq \frac{6\theta_k}{W(\theta_k -\theta_{k+1})}.
\] Also notice that
\[
\left( \frac{\theta_k}{(\theta_k -\theta_{k+1})} \right)^2 \leq w_{k+1} \leq W.
\]

Now by Azuma's inequality, we have
\begin{eqnarray*}
\Pr_{(\pi,S) \sim p}\left[ Z_T \leq -c_1 \right]  &\leq& \exp\left(-\frac{(c_1 - \frac{4T}{W})^2}{2T (\frac{6\theta_k}{W(\theta_k -\theta_{k+1})})^2}\right) \\
&=& \exp\left( -\frac{(c_1-4c)^2(\theta_k-\theta_{k+1})^2W}{72\cdot c \cdot \theta_k^2}\right) \\
&\leq & \exp\left( -\frac{(c_1-4c)^2}{72c}\right).
\end{eqnarray*}

\end{proof}
Finally we are going to use Lemma \ref{lem:bounddiv} with $c_1 = 1/3$. We define $V$ as indicator function of the event $\sum_{i=1}^n \frac{w_i}{W} \ln \frac{p_i(S,\pi)}{p(S,\pi)} \geq -c_1$, i.e.
\begin{enumerate}
\item $V=1$ if $\sum_{i=1}^n \frac{w_i}{W} \ln \frac{p_i(S,\pi)}{p(S,\pi)} \geq -c_1$.
\item $V = 0$, otherwise.
\end{enumerate}

The probability that $A$ identify the top item can be written as
\begin{eqnarray*}
 &&\Pr_{(\pi,S) \sim p} [A(S) = \{ \pi_1,...,\pi_k\}] \\
 &=&  \Pr_{(\pi,S) \sim p} [(A(S) = \{ \pi_1,...,\pi_k\}) \wedge (V=0)]  + \Pr_{(\pi,S) \sim p} [(A(S) = \pi_1) \wedge (V=1)]  \\
 &\leq& \Pr_{(\pi,S) \sim p} [V=0] + \sum_{(\pi,S):A(S)=\{ \pi_1,...,\pi_k\}, V=1} p(S,\pi) \\
 &\leq& \exp\left( -\frac{(c_1-4c)^2}{72c}\right) +  \sum_{(\pi,S):A(S)=\{ \pi_1,...,\pi_k\}, V=1} \left( e^{c_1}  \prod_{i=1}^n p_i(S,\pi)^{\frac{w_i}{W}}\right) \\
 &\leq& \exp\left( -\frac{(c_1-4c)^2}{72c}\right) +  \sum_{(\pi,S):A(S)=\{ \pi_1,...,\pi_k\}, V=1} \left( e^{c_1}  \sum_{i=1}^n \frac{w_i}{W} \cdot p_i(S,\pi) \right)\\
  &\leq& \exp\left( -\frac{(c_1-4c)^2}{72c}\right) +   e^{c_1}  \sum_{i=1}^n \frac{w_i}{W}  \Pr_{(\pi,S) \sim p_i} [(A(S) = \{ \pi_1,...,\pi_k\}) \wedge (V=1)] \\
  &\leq& \exp\left( -\frac{(c_1-4c)^2}{72c}\right) +   e^{c_1}  \sum_{i=1}^n \frac{w_i}{W}  \Pr_{(\pi,S) \sim p_i} [A(S) = \{ \pi_1,...,\pi_k\}] \\
  &\leq& \exp\left( -\frac{(c_1-4c)^2}{72c}\right) +   e^{c_1}  \sum_{i=1}^n \frac{w_i}{W}  \cdot \frac{1}{2} \\
    &\leq& \exp\left( -\frac{(c_1-4c)^2}{72c}\right) + \frac{e^{c_1}}{2} \\
    &\leq& \exp\left( -\frac{(1/3-4c)^2}{72c}\right) +\frac{3}{4}\\
    &\leq&  \frac{1}{8} + \frac{3}{4} = \frac{7}{8}.
\end{eqnarray*}
The last step comes from the fact that $c < 10^{-4}$.
\end{proof}

%!TEX root =  note.tex

The following theorem is very similar to Theorem \ref{thm:lb2ctk}. For some technical reason, it's not very easy to merge the two proofs. But many parts of proofs of these two theorems are very similar.
\begin{theorem}[Restatement of Theorem \ref{thm:lb2}]
\label{thm:lb2app}
Assume $\theta_k > \theta_{k+1}$ and $c < 4\cdot 10^{-4}$. For any algorithm $A$ (can be adaptive), if $A$ uses $c \sum_{ i: i\leq k, \theta_i \leq 2\theta_{k+1}} \frac{\theta_{k+1}^2}{(\theta_{k+1}- \theta_i)^2}$ comparisons of any size (can be $l$-wise comparison for $2 \leq l \leq n$), then $A$ will identify the top-$k$ items with probability at most $7/8$.
\end{theorem}

\begin{proof}
For notation convenience, we set $w_i = \frac{\theta_{k+1}^2}{(\theta_{k+1}- \theta_i)^2}$ for $i$ such that $i \leq k$ and $\theta_i \leq 2\theta_{k+1}$. For other $i$, we set $w_i = 0$. We also set $W = \sum_{i=1}^n w_i$. Then we have $T = c W$.

First of all, we can assume $A$ is deterministic. This is because if $A$ is randomized, we can fix the randomness string which makes $A$ achieves the highest successful probability.

Let $S = (S_1,...,S_T)$ be the history of algorithm. Each $S_t$ is the comparison result of round $t$. Notice that since $A$ is deterministic, with $S_1,...,S_t$, we can determine the labels of items $A$ want to compare in round $t+1$ even when $A$ is adaptive. So there is no point to put the labels of compared items in the history. So we only put the comparison result in the history, i.e $S_t$ is a number in $[n]$ and $S$ is a length-$T$ string of numbers in $[n]$.

Again since $A$ is deterministic, the label $A$ outputs is just a deterministic function of $S$, we use $A(S)$ to denote it. $A$ outputs correctly if $A$ outputs the label of the top-$k$ items, i.e. $A(S) = \{ \pi_1,...,\pi_k\}$.

We use $p(S, \pi)$ to denote the probability that the items are labeled as $\pi$ and $A$ has history $S$. Now consider the case when we set $\theta_i$ equals to $\theta_{k+1}$ for $i \leq k$. In this case the probability of $A(S) =\{ \pi_1,...,\pi_k\}$ should be at most $1/2$ as item $k+1$ and item $i$ have the same weight. We use $p_i(S, \pi)$ to denote the probability that the items are labeled as $\pi$ and $A$ has history $S$ when $\theta_i$ is changed to $\theta_{k+1}$.

Now we prove the following lemma that gives the connection between $p(S,\pi)$ and $p_i(S,\pi)$.

\begin{lemma}
\label{lem:bounddiv2}
Consider $p$ as a distribution over $(\pi,S)$. For all $c_1 > 0$, we have
\[
\Pr_{(\pi,S) \sim p} \left[\left( \sum_{i=1}^n \frac{w_i}{W} \ln \frac{p_i(S,\pi)}{p(S,\pi)} \right) \leq -c_1 \right] \leq \exp\left( -\frac{(c_1-c)^2}{18c}\right).
\]

\end{lemma}
\begin{proof}
Define random variable $Z_t$ to be the following for $t = 1,...,T$ when $(\pi,S)$ is sampled from distribution $p$:
\[
Z_t  = \sum_{i=1}^n \frac{w_i}{W} \ln \frac{p_i(S_1...S_t,\pi)}{p(S_1...S_t,\pi)}.
\]
We have
\[
Z_T = \sum_{i=1}^n \frac{w_i}{W} \ln \frac{p_i(S,\pi)}{p(S,\pi)}.
\]
Now we want to show that sequence $0, Z_1+\frac{1}{W}, ..., Z_t +\frac{t}{W}, ...,Z_T +  \frac{T}{W}$ forms a supermartigale.

Suppose in round $t$, given $S_1,...,S_{t-1}$ and $\pi$, Algorithm $A$ compares items in set $U_t$. Let $\theta_{-i} = \sum_{j \in U_t, j\neq i} \theta_i$. Then we have, with probability $\theta_i / (\theta_i +\theta_{-i})$,
\[
Z_t - Z_{t-1} =  \frac{w_i}{W} \ln\left(  1- \frac{(\theta_i- \theta_{k+1})\theta_{-i}}{(\theta_{k+1}+\theta_{-i})\theta_i}\right) +  \sum_{j \in U_t, j \neq i} \frac{w_j}{W} \ln \left( 1+ \frac{\theta_j - \theta_{k+1}}{\theta_{k+1} + \theta_{-j}} \right)
\]
Here are two simple facts about $\ln$. For $0 \leq x \leq 1$, $\ln(1+x) \geq x - x^2$. For $0 \leq x \leq 1/2$, $ \ln(1-x) \geq -x - x^2$. It's easy to check that for $i$ such that $w_i > 0$, we have $\frac{(\theta_i- \theta_{k+1})\theta_{-i}}{(\theta_{k+1}+\theta_{-i})\theta_i} \leq  1/2$ and $\frac{\theta_i - \theta_{k+1}}{\theta_{k+1} + \theta_{-i}} \leq 1$ . Therefore, by these two facts, for $i$ such that $w_i  >0$, we have
\begin{eqnarray*}
&& \frac{\theta_i}{\theta_i + \theta_{-i}} w_i \ln\left(  1-\frac{(\theta_i- \theta_{k+1})\theta_{-i}}{(\theta_{k+1}+\theta_{-i})\theta_i} \right)  +\frac{\theta_{-i}}{\theta_i + \theta_{-i}} w_i \ln\left( 1+ \frac{\theta_i - \theta_{k+1}}{\theta_{k+1} + \theta_{-i}} \right) \\
&\geq& \frac{ w_i}{\theta_i + \theta_{-i}}\left(  -\frac{(\theta_i- \theta_{k+1})\theta_{-i}}{\theta_{k+1}+\theta_{-i}} - \frac{(\theta_i- \theta_{k+1})^2\theta_{-i}^2}{(\theta_{k+1}+\theta_{-i})^2\theta_i} + \frac{(\theta_i - \theta_{k+1})\theta_{-i}}{\theta_{k+1} + \theta_{-i}} - \frac{(\theta_i - \theta_{k+1})^2\theta_{-i}}{(\theta_{k+1} + \theta_{-i})^2}\right) \\
&=& -w_i \frac{(\theta_{k+1} - \theta_i)^2\theta_{-i}}{(\theta_{k+1} + \theta_{-i})^2 \theta_i}  = - \frac{\theta_{-i}\theta_{k+1}^2}{(\theta_{k+1} + \theta_{-i})^2 \theta_i} \\
&\geq&  -\frac{\theta_{k+1}^2}{(\theta_{k+1} + \theta_{-i}) \theta_i}  \geq-\frac{\theta_i}{\theta_i + \theta_{-i} } .\\
\end{eqnarray*}
The last step comes from the fact that
\[
\theta_{k+1}^2(\theta_i + \theta_{-i}) \leq \theta_i^2 (\theta_{k+1} + \theta_{-i}).
\]
Therefore we have for all $t$ and $S_1,...,S_{t-1}$,
\[
\E[Z_t - Z_{t-1} | S_1,...,S_{t-1}] \geq  -\sum_{i \in U_t} \frac{\theta_i}{W(\theta_i + \theta_{-i})}  \geq -\frac{1}{W}.
\]
As $Z_1,...,Z_{t-1}$ can be determined by $S_1,...,S_{t-1}$, we have for all $t$ and  $Z_1,...,Z_{t-1}$,
\[
\E\left[\left(Z_t  +\frac{t}{W} \right)- \left(Z_{t-1}+\frac{t-1}{W}\right) | Z_1 -\frac{1}{W},...,Z_{t-1} -\frac{t-1}{W}\right] \geq  0.
\]
Therefore sequence $0, Z_1  +\frac{1}{W}, ..., Z_t +\frac{t}{W}, ...,Z_T + \frac{T}{W}$ forms a supermartingale.

Now we want to bound $|Z_t - Z_{t-1}|$. We know that for $0 \leq x \leq 1$, $|\ln (1+x)| \leq x$ and for $0 \leq x \leq 1/2$, $|\ln(1-x)| \leq 2|x|$. Therefore for $i$ such that $w_i > 0$,
\[
|\frac{w_i}{W}\ln\left(  1-\frac{(\theta_i- \theta_{k+1})\theta_{-i}}{(\theta_{k+1}+\theta_{-i})\theta_i} \right)| \leq \frac{w_i}{W} \cdot  \frac{2(\theta_i- \theta_{k+1})\theta_{-i}}{(\theta_{k+1}+\theta_{-i})\theta_i}
\leq
\frac{2\theta_{k+1}^2\theta_{-i}}{W\theta_i(\theta_i -\theta_{k+1})(\theta_{k+1}+\theta_{-i})}
\leq
\frac{2\theta_{k+1}}{W(\theta_k -\theta_{k+1})}
\]
and
\[
|\frac{w_i}{W} \ln\left( 1+ \frac{\theta_i - \theta_{k+1}}{\theta_{k+1} + \theta_{-i}} \right) | \leq\frac{w_i}{W} \cdot \frac{\theta_i - \theta_{k+1}}{\theta_{k+1} + \theta_{-i}} = \frac{\theta_{k+1}^2}{W(\theta_i -\theta_{k+1})(\theta_{k+1}+\theta_{-i})} \leq \frac{\theta_{k+1}}{W(\theta_k -\theta_{k+1})} \cdot \frac{\theta_i}{\theta_i + \theta_{-i}}.
\]
Therefore, we get
\[
|Z_t -Z_{t-1}| \leq\frac{2\theta_{k+1}}{W(\theta_k -\theta_{k+1})} + \sum_{i \in U_t}\frac{\theta_{k+1}}{W(\theta_k -\theta_{k+1})} \cdot \frac{\theta_i}{\theta_i + \theta_{-i}} \leq \frac{3\theta_{k+1}}{W(\theta_k -\theta_{k+1})}.
\] Also notice that
\[
\left( \frac{\theta_{k+1}}{(\theta_k -\theta_{k+1})} \right)^2 \leq w_k \leq W.
\]

Now by Azuma's inequality, we have
\begin{eqnarray*}
\Pr_{(\pi,S) \sim p}\left[ Z_T \leq -c_1 \right]  &\leq& \exp\left(-\frac{(c_1 - \frac{T}{W})^2}{2T (\frac{3\theta_{k+1}}{W(\theta_k -\theta_{k+1})})^2}\right) \\
&=& \exp\left( -\frac{(c_1-c)^2(\theta_k-\theta_{k+1})^2W}{18\cdot c \cdot \theta_{k+1}^2}\right) \\
&\leq & \exp\left( -\frac{(c_1-c)^2}{18c}\right).
\end{eqnarray*}

\end{proof}
After we prove Lemma \ref{lem:bounddiv2}, the rest of the proof is very similar to Theorem \ref{thm:lb2ctk}. We omit the argument.
\end{proof}

\subsection{Lower bounds for arbitrary weights}

%!TEX root =  note.tex

Again, the following theorem is very similar to Theorem \ref{thm:lb2ctk}.
\begin{theorem}[Restatement of Theorem \ref{thm:lb3}]
\label{thm:lb3app}
Assume $c < 1/18$. For any algorithm $A$ (can be adaptive), if $A$ uses $c \sum_{ i: i> k} \frac{\theta_i}{\theta_k}$ comparisons of any size (can be $l$-wise comparison for $2 \leq l \leq n$), then $A$ will identify the top-$k$ items with probability at most $7/8$.
\end{theorem}

\begin{proof}
For notation convenience, we set $w_i = \frac{\theta_i}{\theta_k}$ for $i > k$. For $i \leq k$, we set $w_i = 0$. We also set $W = \sum_{i=1}^n w_i$. Then we have $T = c W$.

First of all, we can assume $A$ is deterministic. This is because if $A$ is randomized, we can fix the randomness string which makes $A$ achieves the highest successful probability.

Let $S = (S_1,...,S_T)$ be the history of algorithm. Each $S_t$ is the comparison result of round $t$. Notice that since $A$ is deterministic, with $S_1,...,S_t$, we can determine the labels of items $A$ want to compare in round $t+1$ even when $A$ is adaptive. So there is no point to put the labels of compared items in the history. So we only put the comparison result in the history, i.e $S_t$ is a number in $[n]$ and $S$ is a length-$T$ string of numbers in $[n]$.

Again since $A$ is deterministic, the label $A$ outputs is just a deterministic function of $S$, we use $A(S)$ to denote it. $A$ outputs correctly if $A$ outputs the label of the top-$k$ items, i.e. $A(S) = \{ \pi_1,...,\pi_k\}$.

We use $p(S, \pi)$ to denote the probability that the items are labeled as $\pi$ and $A$ has history $S$. Now consider the case when we set $\theta_i$ equals to $\theta_k$ for $i > k$. In this case the probability of $A(S) =\{ \pi_1,...,\pi_k\}$ should be at most $1/2$ as item $k$ and item $i$ have the same weight. We use $p_i(S, \pi)$ to denote the probability that the items are labeled as $\pi$ and $A$ has history $S$ when $\theta_i$ is changed to $\theta_k$.

Now we prove the following lemma that gives the connection between $p(S,\pi)$ and $p_i(S,\pi)$.

\begin{lemma}
\label{lem:bounddiv3}
Consider $p$ as a distribution over $(\pi,S)$. For all $c_1 > 0$, we have
\[
\Pr_{(\pi,S) \sim p} \left[\left( \sum_{i=1}^n \frac{w_i}{W} \ln \frac{p_i(S,\pi)}{p(S,\pi)} \right) \leq -c_1 \right] \leq \exp \left( -\frac{(c_1/c -1)^2T}{8}\right).
\]

\end{lemma}
\begin{proof}
Define random variable $Z_t$ to be the following for $t = 1,...,T$ when $(\pi,S)$ is sampled from distribution $p$:
\[
Z_t  = \sum_{i=1}^n \frac{w_i}{W} \ln \frac{p_i(S_1...S_t,\pi)}{p(S_1...S_t,\pi)}.
\]
We have
\[
Z_T = \sum_{i=1}^n \frac{w_i}{W} \ln \frac{p_i(S,\pi)}{p(S,\pi)}.
\]
Now we want to show that sequence $0, Z_1+\frac{1}{W}, ..., Z_t +\frac{t}{W}, ...,Z_T +  \frac{T}{W}$ forms a supermartingale.

Suppose in round $t$, given $S_1,...,S_{t-1}$ and $\pi$, Algorithm $A$ compares items in set $U_t$. Let $\theta_{-i} = \sum_{j \in U_t, j\neq i} \theta_i$. Then we have, with probability $\theta_i / (\theta_i +\theta_{-i})$,
\begin{eqnarray*}
Z_t - Z_{t-1} &=&  \frac{w_i}{W} \ln\left(  1+ \frac{(\theta_k- \theta_i)\theta_{-i}}{(\theta_k+\theta_{-i})\theta_i}\right) +  \sum_{j \in U_t, j \neq i} \frac{w_j}{W} \ln \left( 1- \frac{\theta_k - \theta_j}{\theta_k + \theta_{-j}} \right) \\
&=& -\frac{w_i}{W} \ln\left(  1 -  \frac{(\theta_k- \theta_i)\theta_{-i}}{(\theta_i+\theta_{-i})\theta_k}\right) -  \sum_{j \in U_t, j \neq i} \frac{w_j}{W} \ln \left( 1+ \frac{\theta_k - \theta_j}{\theta_j + \theta_{-j}} \right) \\
\end{eqnarray*}
We are going to use a simple fact about $\ln$: for all $x > -1$, $\ln(1+x) \leq x$.
\begin{eqnarray*}
&&- \frac{\theta_i}{\theta_i + \theta_{-i}} w_i \ln\left(  1-\frac{(\theta_k- \theta_i)\theta_{-i}}{(\theta_i+\theta_{-i})\theta_k} \right)  -\frac{\theta_{-i}}{\theta_i + \theta_{-i}} w_i \ln\left( 1+ \frac{\theta_k - \theta_i}{\theta_i + \theta_{-i}} \right) \\
&\geq& -\frac{ w_i}{\theta_i + \theta_{-i}}\left(-\frac{(\theta_k- \theta_i)\theta_{-i}\theta_i}{(\theta_i+\theta_{-i})\theta_k} +\frac{(\theta_k - \theta_i)\theta_{-i}}{\theta_i + \theta_{-i}}  \right) \\
&=& -\frac{(\theta_k-\theta_i)^2 \theta_{-i} \theta_i}{(\theta_i + \theta_{-i})^2 \theta_k^2} \\
&\geq& -\frac{\theta_i}{\theta_i + \theta_{-i}}.
\end{eqnarray*}
Therefore we have for all $t$ and $S_1,...,S_{t-1}$,
\[
\E[Z_t - Z_{t-1} | S_1,...,S_{t-1}] \geq  -\sum_{i \in U_t} \frac{\theta_i}{W(\theta_i + \theta_{-i})}  \geq -\frac{1}{W}.
\]
As $Z_1,...,Z_{t-1}$ can be determined by $S_1,...,S_{t-1}$, we have for all $t$ and  $Z_1,...,Z_{t-1}$,
\[
\E\left[\left(Z_t  +\frac{t}{W} \right)- \left(Z_{t-1}+\frac{t-1}{W}\right) | Z_1 -\frac{1}{W},...,Z_{t-1} -\frac{t-1}{W}\right] \geq  0.
\]
Therefore sequence $0, Z_1  +\frac{1}{W}, ..., Z_t +\frac{t}{W}, ...,Z_T + \frac{T}{W}$ forms a supermartingale.

Now we want to bound $|Z_t - Z_{t-1}|$. We know that for $0 \leq x$, $|\ln (1+x)| \leq x$. Therefore for $i$ such that $w_i > 0$,
\[
|\frac{w_i}{W} \ln\left(  1+ \frac{(\theta_k- \theta_i)\theta_{-i}}{(\theta_k+\theta_{-i})\theta_i}\right)| \leq \frac{w_i}{W} \cdot  \frac{(\theta_k- \theta_i)\theta_{-i}}{(\theta_k+\theta_{-i})\theta_i} \leq \frac{(\theta_k-\theta_i)\theta_{-i}}{W(\theta_k+\theta_{-i})\theta_k} \leq \frac{1}{W}
\]
and
\[
|\frac{w_i}{W} \ln \left( 1+ \frac{\theta_k - \theta_i}{\theta_i + \theta_{-i}} \right)| \leq\frac{w_i(\theta_k - \theta_i)}{W(\theta_i + \theta_{-i})} = \frac{\theta_i (\theta_k -\theta_i)}{W\theta_k(\theta_i + \theta_{-i})} \leq \frac{1}{W} \cdot \frac{\theta_i}{\theta_i + \theta_{-i}}.
\]
Therefore, we get
\[
|Z_t -Z_{t-1}| \leq\frac{1}{W} + \sum_{i \in U_t}\frac{1}{W} \cdot \frac{\theta_i}{\theta_i + \theta_{-i}} \leq \frac{2}{W}.
\]

Now by Azuma's inequality, we have
\begin{eqnarray*}
\Pr_{(\pi,S) \sim p}\left[ Z_T \leq -c_1 \right]  &\leq& \exp\left(-\frac{(c_1 - \frac{T}{W})^2}{2T (\frac{2}{W})^2}\right) =\exp \left( -\frac{(c_1/c -1)^2T}{8}\right).\\
\end{eqnarray*}

\end{proof}
After we prove Lemma \ref{lem:bounddiv3}, the rest of the proof is very similar to Theorem \ref{thm:lb2ctk} by picking $c_1= 1/3$. We omit the argument.
\end{proof}

%!TEX root =  note.tex

\begin{theorem}[Restatement of Theorem \ref{thm:lb4}]
\label{thm:lb4app}
For any algorithm $A$ (can be adaptive), if $A$ uses $k/4$ comparisons of any size (can be $l$-wise comparison for $2 \leq l \leq n$), then $A$ will identify the top-$k$ items with probability at most $2/3$. 
\end{theorem}

\begin{proof}

First of all, we can assume $A$ is deterministic. This is because if $A$ is randomized, we can fix the randomness string which makes $A$ achieves the highest successful probability. 

Let $S = (S_1,...,S_T)$ be the history of algorithm. Each $S_t$ is the comparison result of round $t$. Notice that since $A$ is deterministic, with $S_1,...,S_t$, we can determine the labels of items $A$ want to compare in round $t+1$ even when $A$ is adaptive. So there is no point to put the labels of compared items in the history. So we only put the comparison result in the history, i.e $S_t$ is a number in $[n]$ and $S$ is a length-$T$ string of numbers in $[n]$.

Again since $A$ is deterministic, the label $A$ outputs is just a deterministic function of $S$, we use $A(S)$ to denote it. $A$ outputs correctly if $A$ outputs the label of the top-$k$ items, i.e. $A(S) = \{ \pi_1,...,\pi_k\}$. 

We use $p(S, \pi)$ to denote the probability that the items are labeled as $\pi$ and $A$ has history $S$. Now consider the case when we set $\theta_i$ equals to $\theta_{k+1}$ for $i\leq k$. In this case the probability of $A(S) =\{ \pi_1,...,\pi_k\}$ should be at most $1/2$ as item $k+1$ and item $i$ have the same weight. We use $p_i(S, \pi)$ to denote the probability that the items are labeled as $\pi$ and $A$ has history $S$ when $\theta_i$ is changed to $\theta_{k+1}$.

We define $N(\pi,S)$ as the set of items among top-$k$ items such that they are not chosen as the favorite items by algorithm $A$ in history $S$ with labels $\pi$. As there are only $k/4$ comparisons, $N(\pi,S) \leq 3/4$ for all $\pi, S$.    

Now we prove the following simple lemma that gives the connection between $p(S,\pi)$ and $p_i(S,\pi)$ for all $i \in N(\pi, S)$.

\begin{lemma}
\label{lem:bounddiv4}
\[
\forall \pi, S, i \in N(\pi, S), p_i(S,\pi) \geq p(S,\pi).
\]
\end{lemma}
\begin{proof}
We write $p(S,\pi)$ as
\[
p(S,\pi) = \prod_{t=1}^T p(S_t,\pi| S_1...S_{t-1}).
\]
And similarly $p_i(S,\pi)$ as 
\[
p_i(S,\pi) = \prod_{t=1}^T p_i(S_t,\pi| S_1...S_{t-1}).
\]
Consider the comparison in round $t$ given $S_1,S_2,...,S_t, \pi$. There are two cases
\begin{enumerate}
\item $i$-th item is not compared in round $t$: The change of $\theta_i$ does not change $p(S_t,\pi| S_1...S_{t-1})$. So $p(S_t,\pi| S_1...S_{t-1}) = p_i(S_t,\pi| S_1...S_{t-1})$.
\item $i$-th item is compared in round $t$: We know the $i$-th item is not the favorite item of round $t$ in this history. Therefore decreasing $\theta_i$ to $\theta_{k+1}$ will increase $p(S_t,\pi| S_1...S_{t-1})$. So $p(S_t,\pi| S_1...S_{t-1}) \leq p_i(S_t,\pi| S_1...S_{t-1})$.
\end{enumerate}
Thus we always have $p(S_t,\pi| S_1...S_{t-1}) \leq p_i(S_t,\pi| S_1...S_{t-1})$. By multiplying things together we get the statement of this lemma.
\end{proof}
Finally we have
\begin{eqnarray*}
 &&\Pr_{(\pi,S) \sim p} [A(S) = \{ \pi_1,...,\pi_k\}] \\
 &=& \sum_{\pi,S, A(S) =  \{ \pi_1,...,\pi_k\}} p(\pi,S) \\
 &\leq& \sum_{\pi,S, A(S) =  \{ \pi_1,...,\pi_k\}} \frac{1}{|N(\pi,S)|} \sum_{i\in N(\pi,S)} p_i(\pi,S) \\ 
 &\leq&\sum_{\pi,S, A(S) =  \{ \pi_1,...,\pi_k\}} \frac{1}{|N(\pi,S)|} \sum_{i\in\{1,...,k\}} p_i(\pi,S) \\ 
 &\leq&\sum_{\pi,S, A(S) =  \{ \pi_1,...,\pi_k\}} \frac{4}{3k} \sum_{i\in\{1,...,k\}} p_i(\pi,S) \\ 
 &=&\frac{4}{3k} \sum_{i\in\{1,...,k\}}  \sum_{\pi,S, A(S) =  \{ \pi_1,...,\pi_k\}} p_i(\pi,S) \\
 &\leq& \frac{4}{3k} \sum_{i\in\{1,...,k\}} \frac{1}{2} \\
 &=&  \frac{2}{3}.
\end{eqnarray*}
\end{proof}

%!TEX root =  note.tex

\begin{theorem}[Restatement of Theorem \ref{thm:lb5}]
\label{thm:lb5app}
Assume $c < 1/2$. For any algorithm $A$ (can be adaptive), if $A$ uses $\frac{c n}{l} $ comparisons of size at most $l$ (can be $2$-wise, $3$-wise,...,$l$-wise comparisons), then $A$ will identify the top-$k$ items with probability at most $7/8$. 
\end{theorem}

\begin{proof}
We are going to prove by contradiction. Suppose there's some $A$ uses $\frac{c n}{l} $ comparisons of size at most $l$ and identify the top-$k$ items with probability more than $7/8$. 

Now consider another task where the goal is just to make sure $k$-th item appeared in some comparison or $(k+1)$-th item appeared in some comparison. Notice that when some algorithm fails this new task, then the algorithm cannot output top-$k$ items with probability better than $1/2$ because when both $k$-th item and $(k+1)$-th item are not compared, the algorithm should output them with same probability for identifying top-$k$ items. So algorithm $A$ should solve the new task with probability more than $3/4$. 

For the new task, it's easy to see that the best strategy is to always use $l$-wise comparison and compare $ \frac{c n}{l} \cdot l$ different items. The probability of having either $k$-th item or $(k+1)$-th item compared is 
\[
1- \left( 1- \frac{2}{n}\right)^{cn} \leq  1- \frac{1}{4^{2c}} \leq 3/4.
\]
Here we need to use the fact that $n\geq 4$ (when $n < 4$, the statement of the theorem is trivial). Now we get a contradiction.
\end{proof}

\section*{Acknowledgement}

We would like to thank Rene Caldentey and Yifan Feng for earlier discussions of this problem.

\bibliographystyle{alpha}
\bibliography{references}

\begin{thebibliography}{CBCTH13}

\bibitem[ACN08]{Ailon08}
N~Ailon, M.~Charikar, and A.~Newman.
\newblock Aggregating inconsistent information: ranking and clustering.
\newblock {\em Journal of the ACM}, 55(5):23:1--23:27, 2008.

\bibitem[Ail11]{Ailon11}
N.~Ailon.
\newblock Active learning ranking from pairwise preferences with almost optimal
  query complexity.
\newblock In {\em Advances in Neural Information Processing Systems}, 2011.

\bibitem[BM08]{Braverman08}
M.~Braverman and E.~Mossel.
\newblock Noisy sorting without resampling.
\newblock In {\em Proceedings of ACM-SIAM symposium on discrete algorithms},
  2008.

\bibitem[BMW16]{BMW16}
M.~Braverman, J.~Mao, and M.~S. Weinberg.
\newblock Parallel algorithms for select and partition with noisy comparisons.
\newblock In {\em Proceedings of the Symposium on the Theory of Computing
  (STOC)}, 2016.

\bibitem[BT52]{Bradley52}
R.~Bradley and M.~Terry.
\newblock Rank analysis of incomplete block designs: I. the method of paired
  comparisons.
\newblock {\em Biometrika}, 39(3/4):324--345, 1952.

\bibitem[BWV13]{Bubeck:13}
S.~Bubeck, T.~Wang, and N.~Viswanathan.
\newblock Multiple identifications in multi-armed bandits.
\newblock In {\em Proceedings of the International Conference on Machine
  Learning (ICML)}, 2013.

\bibitem[CBCTH13]{Chen13}
X.~Chen, P.~N. Bennett, K.~Collins-Thompson, and E.~Horvitz.
\newblock Pairwise ranking aggregation in a crowdsourced setting.
\newblock In {\em Proceedings of the sixth ACM international conference on Web
  search and data mining}, 2013.

\bibitem[CCZZ17]{Chen:17:adaptive}
J.~Chen, X.~Chen, Q.~Zhang, and Y.~Zhou.
\newblock Adaptive multiple-arm identification.
\newblock In {\em Proceedings of International Conference on Machine Learning
  (ICML)}, 2017.

\bibitem[CGMS17]{Chen:17:SODA}
X.~Chen, S.~Gopi, J.~Mao, and J.~Schneider.
\newblock Competitive analysis of the top-$k$ ranking problem.
\newblock In {\em Proceedings of ACM-SIAM Symposium on Discrete Algorithms
  (SODA)}, 2017.

\bibitem[CS15]{Chen15}
Y.~Chen and C.~Suh.
\newblock Spectral {MLE}: Top-k rank aggregation from pairwise comparisons.
\newblock In {\em Proceedings of the International Conference on Machine
  Learning (ICML)}, 2015.

\bibitem[DKNS01]{Dwork01}
C.~Dwork, R.~Kumar, M.~Naor, and D.~Sivakumar.
\newblock Rank aggregation methods for the web.
\newblock In {\em Proceedings of the International World Wide Web Conference},
  2001.

\bibitem[FLN03]{FaginLN03}
Ronald Fagin, Amnon Lotem, and Moni Naor.
\newblock Optimal aggregation algorithms for middleware.
\newblock {\em J. Comput. Syst. Sci.}, 66(4):614--656, 2003.

\bibitem[HSRW16]{Heckel:16}
Reinhard Heckel, Nihar~B. Shah, Kannan Ramchandran, and Martin~J. Wainwright.
\newblock Active ranking from pairwise comparisons and when parametric
  assumptions don��t help.
\newblock arXiv preprint arXiv:1606.08842v2, 2016.

\bibitem[JKSO13]{Jang13}
M.~Jang, S.~Kim, C.~Suh, and S.~Oh.
\newblock Top-$k$ ranking from pairwise comparisons: When spectral ranking is
  optimal.
\newblock arXiv preprint arXiv:1603.04153, 2013.

\bibitem[JMNB14]{Jamieson:14}
K.~Jamieson, M.~Malloy, R.~Nowak, and S.~Bubeck.
\newblock lil' ucb: An optimal exploration algorithm for multi-armed bandits.
\newblock In {\em Proceedings of Conference on Learning Theory}, 2014.

\bibitem[JN11]{Jamieson11}
K.~Jamieson and R.~Nowak.
\newblock Active ranking using pairwise comparisons.
\newblock In {\em Advances in Neural Information Processing Systems}, 2011.

\bibitem[KMS07]{Mathieu07}
C.~Kenyon-Mathieu and W.~Schudy.
\newblock How to rank with few errors.
\newblock In {\em Proceedings of the Symposium on Theory of computing (STOC)},
  2007.

\bibitem[LB11]{Craig:11}
T.~Lu and C.~Boutilier.
\newblock Learning mallows models with pairwise preferences.
\newblock In {\em Proceedings of the International Conference on Machine
  Learning (ICML)}, 2011.

\bibitem[Luc59]{Luce59}
R.~D. Luce.
\newblock {\em Individual choice behavior: A theoretical analysis}.
\newblock New York: Wiley, 1959.

\bibitem[McF73]{McFadden:73}
D.~McFadden.
\newblock Conditional logit analysis of qualitative choice behaviour.
\newblock In P.~Zarembka, editor, {\em Frontiers in Econometrics}, pages
  105--142. Academic Press New York, New York, NY, USA, 1973.

\bibitem[MS16]{Mohajer:16:active}
S.~Mohajer and C.~Suh.
\newblock Active top-k ranking from noisy comparisons.
\newblock In {\em Proceedings of the 54th Annual Allerton Conference on
  Communication, Control, and Computing (Allerton)}, 2016.

\bibitem[NOS17]{Negahban12RankCentrality}
S.~Negahban, S.~Oh, and D.~Sha.
\newblock Rank centrality: Ranking from pair-wise comparisons.
\newblock {\em Operations Research}, 65(1):266--287, 2017.

\bibitem[RA14]{Rajkumar14}
A.~Rajkumar and S.~Agarwal.
\newblock A statistical convergence perspective of algorithms for rank
  aggregation from pairwise data.
\newblock In {\em Proceedings of the International Conference on Machine
  Learning (ICML)}, 2014.

\bibitem[SBGW17]{Shah15Sto}
N.~B. Shah, S.~Balakrishnan, A.~Guntuboyina, and M.~J. Wainright.
\newblock Stochastically transitive models for pairwise comparisons:
  Statistical and computational issues.
\newblock {\em IEEE Transactions on Information Theory (to appear)},
  63(2):934--959, 2017.

\bibitem[SBPH15]{Szorenyi:15}
B.~Sz\"{o}r\'{e}nyi, R.~{Busa-Fekete}, A.~Paul, and E.~H\"{u}llermeier.
\newblock Online rank elicitation for plackett-luce: A dueling bandits
  approach.
\newblock In {\em Proceedings of Advances in Neural Information Processing
  Systems (NIPS)}, 2015.

\bibitem[STZ17]{Suh16Adversarial}
C.~Suh, V.~Tan, and R.~Zhao.
\newblock Adversarial top-$k$ ranking.
\newblock {\em IEEE Transactions on Information Theory}, 63(4):2201--2225,
  2017.

\bibitem[SW15]{Shah15Sim}
N.~B. Shah and M.~Wainwright.
\newblock Simple, robust and optimal ranking from pairwise comparisons.
\newblock arXiv preprint arXiv:1512.08949, 2015.

\bibitem[Tra03]{Train:03:choice}
Kenneth Train.
\newblock {\em Discrete Choice Methods with Simulation}.
\newblock Cambridge University Press, 2003.

\bibitem[WMJ13]{Wauthier13}
F.~Wauthier, M.Jordan, and N.~Jojic.
\newblock Efficient ranking from pairwise comparisons.
\newblock In {\em Proceedings of the International Conference on Machine
  Learning (ICML)}, 2013.

\bibitem[ZCL14]{Zhou:14}
Y.~Zhou, X.~Chen, and J.~Li.
\newblock Optimal pac multiple arm identification with applications to
  crowdsourcing.
\newblock In {\em Proceedings of the International Conference on Machine
  Learning (ICML)}, 2014.

\end{thebibliography}

\end{document}